\documentclass[11pt,twoside]{article}
\usepackage[english]{babel} 
\usepackage[utf8]{inputenc} 
\usepackage[T1]{fontenc}
\usepackage{lmodern}
\usepackage[top = 3cm ,bottom = 3cm,left= 3cm,right = 3cm]{geometry}
\usepackage{amsfonts}
\usepackage{graphicx}
\usepackage{amsmath}
\usepackage{amssymb}
\usepackage{amsthm}
\usepackage{marvosym}
\usepackage{lipsum}
\usepackage{soul,color}
\usepackage{fancyhdr}
\usepackage{authblk}
\usepackage [
n,
advantage,
operators,
sets,
adversary,
landau,
probability,
notions,
logic,
ff,
mm,
primitives,
events,
complexity,
oracles,
asymptotics,
keys
]{cryptocode}
\createprocedureblock{game}{center,boxed}{}{}{}
\usepackage{lipsum}

\newtheorem{theorem}{Theorem}[section]
\newtheorem{proposition}[theorem]{Proposition}
\newtheorem{definition}[theorem]{Definition}

\newcommand{\conc}{\mathbin\Vert}
\newcommand{\cond}{\, | \,}

\usepackage{float}
\floatstyle{ruled} 
\newfloat{protocol}{htb!}{Protocol}

\usepackage{hyperref}
\usepackage{cleveref}
\crefname{protocol}{Protocol}{Protocols}
\Crefname{protocol}{Protocol}{Protocols}

\usepackage{titlesec}
\titleformat{\paragraph}
{\normalfont\normalsize\bfseries}{\theparagraph}{1em}{}
\titlespacing*{\paragraph}
{0pt}{3.25ex plus 1ex minus .2ex}{1.5ex plus .2ex}

\title{A Commitment-based Authentication model for Key Exchange protocols}
\author[3]{\rm Iván Blanco Chacón}
\author[2]{\rm David Domingo Martín}
\author[1]{Ignacio Luengo Velasco}
\author[1,2]{\rm Rodrigo Martín Sánchez-Ledesma \Letter}
\affil[1]{Universidad Complutense de Madrid}
\affil[ ]{\texttt {rodrma01@ucm.es}}
\affil[2]{Indra Sistemas de Comunicaciones Seguras}
\affil[ ]{\texttt {\{rmsanchezledesma, ddomingo\}@indra.es}}
\affil[3]{Universidad de Alcalá}
\affil[ ]{\texttt {ivan.blancoc@uah.es}}
\date{\today}

\begin{document}

\maketitle
\begin{abstract}
In this work we construct an alternative model for Authenticated Key Exchange, intended to build a theoretic security framework for protocols whose characteristics may not always concur with the specifics of already existing models for authenticated exchanges. This model is constructed in a modular way, from the notion of commitment schemes and employing ephemeral information, therefore avoiding the exchange of long-term cryptographic material. From this model, we propose a number of Commitment-based protocols to establish a shared secret between two parties, and study their resistance over unauthenticated channels. This means analyzing the security of the protocol itself, and its robustness against Man-in-the-Middle attacks, by formalizing their security under this  model. The protocols are constructed from Key Agreement (KA) and Key Encapsulation (KEM) primitives, to show that this model can be applied to both established and new paradigms. We highlight the differences that arise naturally, due to the nature of KEM constructions, in terms of the protocol itself and the types of attacks that they are subject to. We provide practical go-to protocols instances to migrate to, both for KEM-based and KA-based cryptographic primitives.
\end{abstract}

\section{Introduction}

Secure key exchange protocols are a fundamental paradigm in cryptography, arguably one of the most important branches of computational mathematics. These protocols allow parties to derive a shared secret from the exchange of public cryptographic material. This secret is then usually used to encrypt the information transmitted between them, by means of  symmetric cryptography \cite{doi:10.1137/S0097539702403773}. It is therefore imperative for these protocols to be secure and provide assurance that, at the end of it, the generated shared secret is only known by the parties involved.

The paradigm of safely exchanging message over unreliable channel has been object of study for a long time: first, trying to model the paradigm, then constructing protocols within this model. Among them, we highlight \cite{10.1007/3-540-44987-6_28,cryptoeprint:1998/009}, which provide a complete solution for modeling the differences between reliable (i.e. authenticated) and unreliable (i.e. unauthenticated) channels, and proving the security of protocols, which include key agreement protocols, in the latter, and how to construct them from protocols on the former.

Currently, secure key exchange protocols are mostly based on KA-type solutions like Diffie-Hellman. In the decade of 1990's, a series of quantum algorithms were published, including Shor's algorithm \cite{Shor_1997}, that were shown to effectively solve the mathematical problems behind most public-key cryptography algorithms, including the discrete logarithm problem or the factorization of integers problem. These critical findings, along with the practical advance on the development of quantum computers and the difficulty and cost of such a migration, forced the cryptography community to start looking for alternatives that are resistant to quantum computing. 

In this scenario, in 2016, the National Institute for Standards and Technologies (NIST) announced the launch of a standardization process, with the objective of coming together with quantum-resistant cryptographic solutions for public-key cryptography needs, specifically suitable digital signature and key establishment algorithms. With regards to key establishment, the call specified that looked for Key Encapsulation Mechanisms (KEM) schemes, which not only are efficient and versatile but also easy to derive from Public Key Encryption (PKE) algorithms, based on transformations such as Fujisaki-Okamoto and variants \cite{10.1007/978-3-319-70500-2_12,10.1007/978-3-540-40974-8_12}.

Based on the wide spread of KA-based protocols, it is important that quantum-resistant KEM schemes are as little disruptive as possible, to facilitate their early adoption, in an attempt to minimize protected information that becomes vulnerable once \textit{Cryptographically Relevant Quantum Computers} (CRQC) become available.

\subsection{Motivation}
Regarding secure key exchange protocols over unauthenticated channels, the pioneer works of \cite{10.1007/3-540-44987-6_28,cryptoeprint:1998/009} have been thoroughly studied and applied, and subsequent important works have been derived after, like \cite{10.1007/11535218_33,10.1007/978-3-540-75670-5_1}. For protocols based on quantum-resistant primitives, there already exists works in the same direction. In \cite{10.1007/978-3-031-35486-1_24}, the authors expanded the unauthenticated model of \cite{10.1007/3-540-44987-6_28}, to consider authentication via KEM schemes, a critical step towards post-quantum secure protocols over unauthenticated channels.

In the original unauthenticated model of \cite{cryptoeprint:1998/009}, the authors present a model in which authentication is achieved by an initial authenticated exchange of cryptographic keys between the users which intend to communicate. Then, a number of techniques can be applied from these keys to provide the necessary authentication, even when the subsequent communications are performed over an insecure channel. 

This model builds a framework to construct secure and authenticated protocols to be performed over unauthenticated channels, if the steps of the model are followed. This means that an initial, out-of-bound authenticated exchange of long-term cryptographic keys is required. 

While this requirement can be appropriate for a vast number of protocols, it could pose a steep restriction in others, for example, among real-time secure communication protocols, in which the notion of a previous authenticated exchange between parties might not be possible, specially when this notion of authenticated cryptographic exchange usually involves present exchange between peers.

The present work addresses this situation by building an alternative model, in which authentication is achieved not through an initial authenticated exchange of cryptographic material but rather of a final authenticated exchange of (not  necessarily cryptographic) information derived from the protocol session established. This value would corroborated between parties before the communication starts, in an authenticated way.

In comparison, our model requires fewer public key operations to achieve authentication while eliminating the need for long-term secrets, at the expense of requiring and authentication exchange in each session established between peers, as opposed to a single cryptographic material exchange. 

More so, the definition of this model can serve as a formal security framework for a number of practical solutions which may already be employing the techniques defined in this work, or similar notions, which can fall within the scope of our model. Furthermore, it could be considered as the way to achieve the initial authenticated cryptographic key exchange phase required by the original AKE model.

\subsection{Our contributions}

The first objective of this work is to derive a new authentication model for secure protocols, based purely on ephemeral elements, avoiding the need for long-term cryptographic material. 

To this end, we first construct a variation of the AKE model presented on \cite{cryptoeprint:1998/009}, which varies the way authentication is achieved, reducing the number of public key operation needed to obtain it. This authentication will be obtained by a tampering verification of the elements exchanged, once the exchange run is finished. The measures employed are expected to be simple and provide easy adoption. 

Most importantly, we will show that all properties and results obtained under the original model translates to our model, including those pertaining to key exchange protocols defined on \cite{10.1007/3-540-44987-6_28}. This model will be the base to generate and formally verify the protocols that will appear on this work.

Due to the menace of \textit{quantum computing}, it is also the objective of this work to include KEM schemes into consideration, and to define KEM-based protocols who are easy to migrate to from KA-based analogues.

To this end, we analyze a number of KA-based protocols for secure key establishment, breaking them by the number of messages exchanged. We show that applying the results from our AKE model, the protocols are shown to be secure under unreliable channels. We will also analyze how the specific nature of Key Encapsulation Mechanism forces more difficulties in order to come up with secure KEM-based key establishment protocols, and the measures needed to ensure their security. Lastly, we will specify a number of KEM-based protocols, using the same classification as above and their formal security under our new model.

Special focus will be put around the figure of a MitM, whose aim is to eavesdrop the communication between the two parties, by means of modifying the messages exchanged during the key establishment protocol, in a way its presence goes unnoticed. We will analyze the possibilities of such an attacker in KA-based protocols and how to detect and avoid them, which will conform the base of the authentication measures of our model. Then we will see that the measures applied in this setting do not translate naturally to KEM-based protocols, under the same conditions, and we will detail what specific measures are required to block this kind of attacks, depending on the inherent conditions of the KEM in use. 

Despite the conceptual differences between the two paradigms, the measures under which our model will base the security over unreliable channels will be the same: along with the primary objective, which is to derive a shared secret key, the protocol will also generate another output, which will be referred to as \textit{session Authentication Value (AV)}. This value will represent a summary of the session between the two parties, as it will be directly inferred by the shared values used throughout the process. This value will determine whether the session exchange has been successful or it has been tampered.

It is important to note that this model does not attempt to substitute the original, well-established authentication model originally presented in \cite{10.1007/3-540-44987-6_28}. This model is the standard paradigm in a great number of practical applications and protocols, all of which  can assume the presence of an out-of-bound authenticated exchange of cryptographic material.

There are scenarios, however, in which this assumption might pose a requirement too big to be satisfied. It is under these scenarios we believe our model could be of assistance, as it makes no use of long-term cryptographic material to be authenticated. At the very least, this model helps capture the intuition behind purely ephemeral-based approaches to AKE, and the kind of tools they might require. 

The protocols presented throughout this work rely on an out-of-bound final authentication phase, denoted by $I_f$ and detailed below. Intuitively, under this phase the parties involved in a protocol run will verify, in an authenticated way, whether the session AV values generated by each of them in the protocol run match.

As the original model stated with the initialization phase $I$, this work does not seek to formalize how $I_f$ should look, beyond the requirement that $I_f$ conform an authenticated verification by the receiver, of the AV generated in the protocol run, against the sender's one.
The particularities of how this phase is achieved will be dependent on each practical application of the model and are outside of the scope of this work. 

Nevertheless, for completeness, some examples could include: the use of biometric indicators or other identity-authentication tools, the presence of trusted third party identity providers, or the existence of authenticated channels only suitable for small chunks of information. Any of these could be used as an authenticator of the information to be validated, which could be represented as short throughput of information. This notion of short value authentication is already in use in a number of secure communications protocols, thus our modified model provides a theoretic security framework for these kind of protocols.

Another key difference between the model presented in \cite{10.1007/3-540-44987-6_28} and this one is the nature of the elements employed within their respective authentication phases. The elements to be authenticated under the model presented in \cite{10.1007/3-540-44987-6_28} (i.e. the cryptographic public keys) are purely of cryptographic nature. As such, they require a specific information processing that cannot be limited to human-verifiable conditions. In other words, these elements are needed by the underlying cryptographic algorithms employed within the protocol run.

On the other hand, the elements exchanged under the model presented in this work are not of cryptographic nature. This is because the session AVs generated are just information that needs to be verified, and can be represented in different ways to accommodate $I_f$. Therefore, it can account for human-verifiable elements of authentication that do not require to be processed by the underlying protocol. This difference also opens up the possibilities for the way of exchanging the information.

\section{Preliminaries}

\begin{definition}
A public-key encryption scheme PKE = (KGen, Enc, Dec) consists of three algorithms and a
finite message space M. The key generation algorithm KGen outputs a key pair $(\pk, \sk)$, where pk also defines a randomness space R = R($\pk$). The encryption
algorithm Enc, on input $\pk$ and a message $\textit{m} \in M$, outputs an encryption c $\sample$ Enc($\pk$, \textit{m}) of \textit{m} under the public key $\pk$. If necessary, we make the used randomness of encryption explicit by writing \textit{c} := Enc(pk, \textit{m}; r), where r $\sample$ R and R is the randomness space. The decryption algorithm Dec, on input $\sk$ and a ciphertext \textit{c}, outputs either a message \textit{m} = Dec($\sk$, \textit{c}) $\in$ M or a special symbol $\bot \notin $ M to indicate that \textit{c} is not a valid ciphertext.
\end{definition}
\begin{definition}
A key encapsulation mechanism is a triple of algorithms KEM = (KGen, Encaps, Decaps). The key generation algorithm KGen generates a key pair $(\pk, \sk)$. The encapsulation algorithm Encaps, given a public key value $\pk$, outputs the pair (ct, K), where \textit{ct} is called the encapsulation of a random value \textit{x} that determines the shared key K. The deterministic decapsulation algorithm Decaps, given the secret key $\sk$ and the encapsulation, outputs the same key K by extracting the random value from the encapsulation \textit{ct}.
\end{definition}
\begin{definition}
(Indistinguishability against KEM scheme) We define the IND-\textit{atk} game, $atk \in \{CPA, CCA\}$ as in the figure below and the IND-\textit{atk} advantage of an adversary $\mathcal{A}$ against the above KEM scheme as 
\begin{align*}
    \advantage{\texttt{IND-atk}}{\texttt{KEM}}[(\mathcal{A})] := 2 \cdot\abs{\prob{\texttt{IND-atk$^{\mathcal{A}}$} \Rightarrow 1} - \frac{1}{2}}
\end{align*}

\begin{pchstack}[boxed, center, space=1em]
  {\procedure[linenumbering]{$\indcpa^\adv$ $\pcbox{\indcca^\adv}$}{\phantomsection\label{kemgame}
      (\pk,\sk) \sample \texttt{KGen}  \\
      b \sample \bin  \\
      (ct^*, K_0^*) \sample \texttt{Encaps}(\pk)  \\
      K_1^* \sample \textit{K}  \\
      b' \sample \adv^{\pcbox{Decaps}}(\pk, ct^*, K_b^*) \\
      \pcreturn b = b'
    }}

  {\procedure[linenumbering] {Oracle $Decaps(ct)$}{%
        \text{if ct = ct$^*$} \\
        \text{ return $\bot$} \\
        else \\
        \text{return \texttt{Decaps}(sk,ct)}
   }}
\end{pchstack}
\end{definition}

\begin{definition}
A Commitment Scheme is a triple (Setup, Com, Open) such that:
\begin{itemize}
    \item $CK \sample \texttt{Setup}(\secparam)$ generates the public commitment context. It is often omitted mentioning the public context CK when clear.
    \item for any $ m \in M$, message space, $(c, d) \sample \texttt{Com}$($m$) is the commitment/opening pair for $m$. $c$ = $c(m)$ serves as the commitment value, and $d$ = $d(m)$ as the opening value.
    \item Open($c$, $d$) = $m'\in M \cup \{\bot\}$, where $\bot $ is returned if c is not a valid commitment to any message.
\end{itemize}
We define the hiding and binding games and advantages of an adversary $\mathcal{A}$ and security parameter \textit{n} as follows:
\begin{pchstack}[boxed, center, space=1em]
      {\procedure[linenumbering]{Hiding$^\adv(n)$}{\phantomsection\label{Hiding}
      CK \sample \texttt{Setup}(\secparam) \\
      (x_0, x_1) \sample \adv_1 (\secparam)  \\
      b \sample \bin  \\
      (c(x_b), d(x_b)) \sample \texttt{Com}(x_b)  \\
      b' \sample \adv_2(c(x_b)) \\
      \pcreturn b = b'
    }}
      {\procedure[linenumbering]{Binding$^\adv(n)$}{\phantomsection\label{Binding}
      CK \sample \texttt{Setup}(\secparam) \\
      (c, d, d') \sample \adv (\secparam)  \\
      m = Open(c, d)  \\
      m' = Open(c, d')  \\
      \pcreturn (m \neq m') \wedge (m, m' \neq \bot)
    }}
\end{pchstack}
\begin{center}
    \begin{align*}
        \advantage{\texttt{Hiding}}{\texttt{CS}}[(\mathcal{A}, n)] := 2 \cdot\abs{\prob{\texttt{Hiding$^\adv$} \Rightarrow 1} - \frac{1}{2}} \leq \text{negl(n)} \\
        \advantage{\texttt{Binding}}{\texttt{CS}}[(\mathcal{A}, n)] := \prob{\texttt{Binding$^\adv$} \Rightarrow 1} 
        \leq \text{negl(n)} \\
    \end{align*}
\end{center}
\end{definition}

\begin{definition}
A Robust Commitment Scheme (RCS) is defined to be a Commitment Scheme such that it verifies the following additional security properties:
\begin{pchstack}[boxed, center, space=1em]
      {\procedure[linenumbering]{Strong-Binding$^\adv(n)$}{\phantomsection\label{Strong-Binding}
      CK \sample \texttt{Setup}(\secparam) \\
      (c, d, d') \sample \adv (\secparam)  \\
      m = Open(c, d)  \\
      m' = Open(c, d')  \\
      \pcreturn (m = m') \wedge (d' \neq d) \wedge (m \neq \bot)
    }}
      {\procedure[linenumbering]{CR$^\adv(n)$}{\phantomsection\label{CR-CS}
      CK \sample \texttt{Setup}(\secparam) \\
      (c, c', d) \sample \adv (\secparam)  \\
      m = Open(c, d)  \\
      m' = Open(c', d)  \\
      \pcreturn (m, m' \neq \bot) \wedge (c' \neq c)
    }}
\end{pchstack}
\begin{center}
    \begin{align*}
        \advantage{\texttt{Strong-Binding}}{\texttt{CS}}[(\mathcal{A}, n)] := \prob{\texttt{Strong-Binding$^\adv$} \Rightarrow 1} 
        \leq \text{negl(n)} \\
        \advantage{\texttt{CR}}{\texttt{CS}}[(\mathcal{A}, n)] := \prob{\texttt{CR$^\adv$} \Rightarrow 1} 
        \leq \text{negl(n)} \\
    \end{align*}
\end{center}
\end{definition}

\begin{definition}
    A cryptographic hash function (CHF) is a function $H: \{0,1\}^* \rightarrow \{0,1\}^n$, which takes arbitrary length inputs and returns outputs of fixed length $n$. Furthermore, they verify:
    \begin{itemize}
        \item Preimage resistance: Given $y \in \{0,1\}^n$, it is hard to find $x$ such that $H(x) = y$.
        \item Second preimage resistance: Given an input $x$, it is hard to find $x' \neq x$ such that $H(x) = H(x')$.
        \item Collision resistance:  It is hard to find two inputs $x \neq x'$ such that $H(x) = H(x')$.
    \end{itemize}
\end{definition}

Note that, in the above definition, it remains to define what \textit{hard} means. For the sake of this definition, we will consider that to mean that it requires efforts comparable to a brute-force search, i.e. $O(2^n)$. This definition implies that, if $n$ is chosen to be small enough, this wont actually be 'hard' in the computational sense.

Throughout this work, we will model the random functions employed in the commitment model as random oracles. This is done to simplify the proofs when considering the probability of a random value $x$ to have a certain value $H(x) = y$, which under the ROM will be $2^{-n}$. We make the convention that $H(x_1 \conc ... \conc x_n) = \bot$ if any $x_i = \bot$.

Moreover, as per \cite{10.1007/978-3-319-70500-2_12}, we make the convention that, in order to keep
record of the queries issued to $H$, we will use a hash list $L_H$ that contains all tuples $(x,H(x))$ of arguments that H was queried on and the respective answers $H(x)$. This way, we can ensure that the number $q$ of queries the adversary makes are different, and that the random oracle is deterministic.

\begin{pchstack}[boxed, center, space=1em]
  {\procedure[linenumbering] {Oracle $H(m)$}{%
        \text{if $\exists r$ s.t $(m, r) \in L_H$} \\
        \text{ return $r$} \\
        r \sample \{0, 1\}^N \\
        L_H = L_H \cup \{(m, r)\} \\
        \text{return $r$}
   }}
\end{pchstack}

Now, we define the central element of the authentication security of our model to give a proper definition for the element that will be used to detect session interference by an attacker:.
\begin{definition}\label{AV_G}
We define the \textit{Authentication Value (AV)} as a deterministic digest of shared elements involved within a protocol key establishment run. Formally, AV are constructed as
\begin{align*}
    G(A_1, ..., A_j) := G(enc(A_1) \conc ... \conc enc(A_j))
\end{align*}
where $A_i$ , $i \in \{1, .., j\}$ are elements known to both parties at the end of the protocol session, G is a CHF, and $enc$ is an encoding that ensures that collisions cannot happen between inputs. In other words, $enc(x_1) \conc enc(y_1) = enc(x_2) \conc enc(y_2) \iff (x_1, y_1) = (x_2, y_2).$

For the purpose of this work, we will define $enc(A_j)$ as $j \conc len(A_j) \conc A_j$, where $len(x)$ is the length operator. This represents a standard Type-Length-Value (TLV) encoding.
\end{definition}
Within each protocol, the specific values that will conform this AV session value will be established.

A compromise must be reached in terms of the actual output length $n$ of the CHF used to derive the \textit{AV} values. Due to practical uses regarding real-time authentication, one would desire its length to be as small as possible, in order to facilitate usability. On the other hand, the bigger this length value is, the higher the protection it provides, as increases the computational complexity of the attacks against it.

We also define an interesting property that will appear repeatedly throughout the work, defined as the \textit{Combined Collision} advantage of an adversary against a CHF:
\begin{definition}[Combined collision advantage] Let $l \in \mathbb{N}$ and $\mathcal{Y}$ be an algorithm that returns $y \in Y$ with a certain probability and $\bot$ otherwise. We define the Combined Collision game and the Combined Collision advantage of a (possibly unbounded) adversary $\mathcal{A}$ against a CHF \textbf{G} as 
\begin{align*}
    \advantage{\texttt{combined}}{\texttt{CHF}}[(\mathcal{A}, l, \mathcal{Y})] := \prob{\texttt{combined$^{\mathcal{A}}(l, \mathcal{Y})$} \Rightarrow 1}
\end{align*}

\begin{pchstack}[boxed, center, space=1em]
  {\procedure[linenumbering]{$combined^\adv(l, Y)$}{\phantomsection\label{combined_game}
      T \sample \bin^l \\
      y \sample \adv^{G, \mathcal{Y}}(T) \\
      \pcreturn (G(y) = T) \wedge (T \neq \bot) \wedge (y \in Y \cup \{\bot\})
    }}
\end{pchstack}
where $l$ represents the output length of G and Y represents the domain from which $y$ must be generated.
\end{definition}
The advantage of an adversary against the above game will be critical to the (in)security of the protocols presented in the following sections. In many cases, this game will model the actual advantage of an adversary in forging the authentication.

The following results provides more details regarding the actual advantage against the above game:
\begin{proposition}\label{combined_adv}
    For any adversary $\adv$ capable of at most $q$ queries to a CHF random oracle $\mathcal{O}$, it holds that
    $$\advantage{\texttt{combined}}{\texttt{CHF}}[(\mathcal{A}, l, \mathcal{Y})] \leq \frac{q}{2^l} \cdot \delta$$
    where $l$ represents the output length of $\mathcal{O}$, and $\delta := \advantage{\texttt{sample}}{\mathcal{Y}}[(\mathcal{A})]$ represents the probability of algorithm $\mathcal{Y}$ to return $y \in Y$.
\end{proposition}
\begin{proof}
Let \textit{F} be the event that an adversary $\adv$ wins the game above. This happens if the adversary is able to find an $y \in Y$ such that $O(y) = T$. Note that the requisite $y \neq \bot$ comes directly from $T \neq \bot$.

Since the adversary is limited to, at most, $q$ (distinct) queries to $\mathcal{O}$, then at most $q$ different values $y_i$ can be consulted. Therefore, define $F_i$ as the event of the $i$-th (distinct) test query to be correct, i.e., the $i$-th element $y_i \in Y_i := Y \setminus \{y_1, \cdots, y_{i-1}\}$ is consulted to the random oracle yields a collision $O(y_i) = T$. 

Then,
\begin{align*}
    & \prob{F_i} := \prob{y_i \in Y_i \wedge O(y_i) = T \cond \neg F_1, ..., \neg F_{i-1}} = \\
    & \prob{O(y_i) = T \cond y_i \in Y_i \cap \neg F_1, ..., \neg F_{i-1}} 
    \cdot \prob{y_i \in Y_i \cond \neg F_1, ..., \neg F_{i-1}} \leq \frac{\delta}{2^l}
\end{align*}
as 
\begin{align*}
    & \prob{y_i \in Y_i \cond \neg F_1, ..., \neg F_{i-1}} = \advantage{\texttt{sample}}{\mathcal{Y}}[(\mathcal{A})] \leq \delta \\
    & \prob{O(y_i) = T \cond y_i \in Y_i \cap \neg F_1, ..., \neg F_{i-1}} = \frac{1}{2^l}
\end{align*}
because we cannot ensure that algorithm $\mathcal{Y}$ will return $y_i \neq y_j$, for some $j < i$.

Therefore, applying the union bound on the events $F_1, ..., F_q$, we arrive to the desired result. Note that this event is defined as is because, to succeed, at one point the adversary has to consult to the oracle with values $y_i \in Y \cup \{\bot\}$.
\end{proof}

Now, if we assume that an adversary is able to come with a set of different elements $y_i \in Y$, then we can have an exact measure of the hardness of the above problem
\begin{proposition}\label{combined_adv_with_Y}
    For any adversary $\adv$ capable of at most $q$ queries to a CHF random oracle $\mathcal{O}$, given access to a set $Y_p$ of distinct values $y_i \in Y$, it holds that
    $$\advantage{\texttt{combined}}{\texttt{CHF}}[(\mathcal{A}, l, \mathcal{Y})] = 1 - \left(1 - \frac{1}{2^l}\right)^{\min{(q, |Y_p|)}}$$
    where $l$ represents the output length of $\mathcal{O}$.
\end{proposition}
\begin{proof}
The probability of failure means that each of the distinct candidates in $Y_p$ do not yield a satisfactory collision. For a single query, this means
\begin{align*}
    & \prob{\mathcal{O}(y_i) \neq T} = 1 - \frac{1}{2^l}
\end{align*}
and since the oracle outputs are independent and the values in $Y_p$ are distinct, we have that
\begin{align*}
    & \prob{\text{no coll}} := \bigcap_{i = 1}^q \prob{\mathcal{O}(y_i) \neq T} = \left(1 - \frac{1}{2^l}\right)^q
\end{align*}
as we can only make $\min{(q, |Y_p|)}$ evaluations, depending upon the maximum number of queries $q$ of the adversary and the cardinal of the set $Y_p$.

Therefore, 
$$\advantage{\texttt{combined}}{\texttt{CHF}}[(\mathcal{A}, l, \mathcal{Y})] = 1 - \prob{\text{no coll}} = 1 - \left(1 - \frac{1}{2^l}\right)^{\min{(q, |Y_p|)}}$$
\end{proof}

\begin{definition} (Decisional Diffie-Hellman Assumption) \cite{Canetti2011} Given a finite, cyclic group $G$ with generator $g$, and three elements, the Decisional Diffie–Hellman (DDH) problem is to decide whether there exist integers $x$, $y$ such that $a = g^x$, $b = g^y$, and $c = g^{xy}$. The DDH assumption states that this problem is hard to solve for some infinite sequence of groups 
of increasing order.
\end{definition}

\section{Authenticated Key Exchange model variation}

The main dilemma that this work and others before try to deal with is the execution of safe communication protocols over non reliable channels. This means, they are carried on under channels where messages might not only be delayed or lost, but also actively modified or even completely substituted. 

Therefore, is it necessary to come with ways of ensuring that no message can be modified, i.e., to provide assurance that what is received by a party is what the other party actually intended. Over key exchange protocols, this amounts to ensure that a safe and secret shared key has been established between parties, without any knowledge of it by any other party not intended.

\subsection{CK model}

Among the cryptographic literature pertaining to this problematic, the work of \cite{10.1007/3-540-44987-6_28} and \cite{cryptoeprint:1998/009} is highlighted. In it, a theoretic modeling of this situation is first described. Then, a number of techniques are defined to construct secure protocols, and formally prove their security under this model, over unauthenticated channels. The following definitions are relevant to the model:
\begin{definition} \cite{10.1007/3-540-44987-6_28}: \textbf{Message-driven protocols} are collections of interactive procedures, run on currently by parties, that specify a particular processing of incoming messages and the generation of outgoing messages. Protocols are initially triggered at a party by an external 'call' and later by the arrival of messages. Upon each of these events, and according to the protocol specification, the protocol processes information and may generate and transmit a message and/or wait for the next message to arrive.
\end{definition}

\begin{definition}
\cite{10.1007/3-540-44987-6_28}: Protocols can trigger the initiation of sub-protocols (i.e. interactive subroutines) or other protocols, and several copies of such protocols may be simultaneously run by each party. Each copy of a protocol run at a party is defined a \textbf{session}. Technically, a session is an interactive subroutine executed inside a party. Each session is identified by the party that runs it, the parties with whom the session communicates and by a session-identifier. These identifiers are used in practice to bind transmitted messages to their corresponding sessions.
\end{definition}
Regarding the theoretic model construction, first the \textit{Authenticated Model} (AM) is defined: Under this model, any attacker willing to confront a set of parties $P_1, .., P_n$ is able to delay and hold onto messages, but is forced, for every message that is delivered, to be done in good faith, without any modification to it. 

When a message is sent by a party, the message is added to a set $M$ of undelivered messages. When a party receives a message, this message is deleted from $M$ (therefore, since the adversary is restricted to deliver faithful messages, if a party $P_j$ receives a message $m$ from a party $P_i$, this must mean that $P_i$ sent message $m$ to $P_j$). 

The attacker is also able to corrupt any party at wish. This means, learning all information (public and private) pertaining to it. If a party is corrupted, then the attacker is allowed to impersonate them, but such action is recorded accordingly in the party's protocol output. The global output of the protocol is the concatenation of the output of each party at the end of it, plus the output of the attacker. 

Each party's output is engrossed by every activation regarding the protocol flow. The adversary's output is defined to be '\textit{all the information seen (and derived) by the adversary throughout the computation, together with its random input. This includes registration of important events that occurred during the execution, such as corruption of parties}'.

Then, the \textit{Unauthenticated Model} (UM) is defined: Under this model, an attacker is not constrained to faithful delivery of message. Messages can be partially modified or even forged altogether. The key component is that, under this model, the protocol is augmented with what is called an \textit{Initialization function} $I$, which models an '\textit{initial phase of out-of-band and authenticated information exchange between the parties}'. This component is essential to the UM model, as seeks to provide the necessary authentication source that is lost once migrating to insecure channels.

In all the literature that works upon similar ideas, this initialization function is used to safely exchange public or secret information, in a way that the information exchanged conforms the seed in which authentication is build upon. \\
Now, it is require to define what does it mean for two protocols to be equivalent, under these models:
\begin{definition} \cite{10.1007/978-3-031-35486-1_24}\label{AM-UM-DistExp}:
The AM-UM distinguishing experiment, $G^{AM-UM-dist}_{\Pi_{AM}-\Pi_{UM}}(\mathcal{D})$ proceeds as follows:
\begin{enumerate}
    \item A uniform bit b $\in \{0, 1\}$ is chosen. If b = 0, $\mathcal{D}$ will interact with the AM protocol $\Pi_{AM}$ and an adversary in it. Otherwise, $\mathcal{D}$ will interact with the UM protocol $\Pi_{UM}$ and an adversary in it.
    \item To conclude the experiment, $\mathcal{D}$ will halt and output b'.
    \item The experiment will output 1 if and only if b = b'.
\end{enumerate}
We define the advantage of the distinguisher $\mathcal{D}$ to be
    \begin{align*}
        \advantage{\textnormal{AM-UM-dist}}{\Pi_{AM}-\Pi_{UM}}[(\mathcal{D})] = 2\cdot\abs{P[G^{AM-UM-dist}_{\Pi_{AM}-\Pi_{UM}}(\mathcal{D}) = 1] - \frac{1}{2}}
    \end{align*}
\end{definition}

\begin{definition} \cite{10.1007/978-3-031-35486-1_24} \label{am-um-emulate}: 
Let $\Pi_{AM}$ and $\Pi_{UM}$ be message-driven protocols for $n$ parties. We say that $\Pi_{UM}$ $\epsilon-$emulates $\Pi_{AM}$ over unauthenticated networks if, for any UM-adversary \textit{U}, there is a AM-adversary \textit{A} such that, for each distinguisher $D$ playing the AM-UM distinguishing game, we have
\[\advantage{\textnormal{AM-UM-dist}}{\Pi_{AM}-\Pi_{UM}}[(\mathcal{D})] \leq \epsilon\]
This means, their respective global outputs are computationally indistinguishable.
\end{definition}

Then, the notion of a compiler is defined, to provide a bridge between protocols in both models:

\begin{definition} \cite{cryptoeprint:1998/009} \label{def_compiler}: A compiler C is an algorithm that takes for input descriptions of protocols and outputs descriptions of protocols. An authenticator is a compiler C where for any protocol $\pi$, the protocol C($\pi$) emulates $\pi$ over unauthenticated networks.
\end{definition}

\subsubsection{Key exchange protocols and SK security} \label{or_sk_security}
A key exchange protocol is a special type of the message-driven protocol described above, introduced in \cite{10.1007/3-540-44987-6_28}: the local output of said protocols are of the form ($P_i$, $P_j$, $s$, $\kappa$), representing the parties involved, the session identifier $s$ and the key derived from the execution of it. Regarding the status of a session, it is considered that if the key value $\kappa$ is null, then the exchange was unsuccessful and is considered an aborted session. Otherwise, the session is considered completed and $\kappa$ is stored as the \textit{secret} of the exchange. The state of such sessions consist of:
\begin{itemize}
	\item The status of the SK-session: completed, aborted or in process.
	\item Any intermediate value generated from the protocol execution, and relevant to complete it.
	\item The session key $\kappa$ derived from the exchange.
\end{itemize}
A complete definition of the concept of \textit{SK security} also needs to be introduced. SK security, intuitively, is about adding to any adversary the possibility to \textit{test} a session key at any given time during the execution of a simulation. 

In it, the challenger will draw a random bit and, depending on it, the adversary will either learn the true value of the key derived from the session they chose (among completed, uncorrupted sessions) or a value drawn at random. The challenge to the adversary is being able to correctly distinguish between them. 

The security of the exchange is therefore tied to any adversary's possibility of distinguishing. Formally, given the following adversary's possibilities:
\begin{itemize}
    \item \textit{NewSession(A, B, s, r)}: the adversary issues the NewSession query to party $A$, specifying its intended receiver $B$, the session identifier $s$, and the role $r$ (initiator or responder) of $A$ in the session. $A$ will follow the protocol definition and may return an output message intended for $B$.
    \item \textit{Send(A, B, m)}: represents activation of $A$ by an incoming message $m$ (possibly including a session identifier) from party $B$. $A$ will follow the protocol and may reject, accept, or return an output message intended for $B$.
    \item \textit{Corrupt(A)}: the adversary learns the whole state of $A$, including any information necessary to impersonate $A$ in an authenticated way (which may pertain to learn any long-term secret or other information, depending on each particular model). The corruption event is recorded in the local output of $A$. Subsequently $A$ can never be activated but the adversary can take the role of $A$ in the protocol.
    \item \textit{RevealKey(A, B, s)}: the adversary learns the session key accepted in the session $s$ by $A$ with party $B$, if it exists. The reveal event is recorded in the local output of $A$.
    \item \textit{RevealState(A, B, s)}: the adversary learns the state information associated with session $s$ at $A$, such as ephemeral keys. The reveal state event is recorded in the local output of $A$.
    \item \textit{Expire(A, B, s)}: if there is a completed session $s$ at $A$ with $B$ then any session key associated with that session is deleted from the memory of $A$. The \textit{Expire} event is recorded in the local output of $A$.
    \item \textit{Test(A, B, s)}: this query can be asked only once and can only be made to a completed session $s$ at $A$ with party $B$. Furthermore there cannot have been any of the following queries made: \textit{RevealKey(A, B, s)} or \textit{RevealState(A, B, s)} or \textit{Corrupt(A)} or \textit{Corrupt(B)}. If the bit $b$ specified by the challenger is $b = 1$ then the session key is returned. Otherwise $b = 0$ and a random key from the keyspace is returned.
\end{itemize}
we define the SK-security of a protocol as follows:
\begin{definition} \cite{10.1007/978-3-031-35486-1_24} The definition of the SK-exchange indistinguishability experiment $ G^{Key-Ind}_{\Pi}(\mathcal{A})$ is as follows
    \begin{enumerate}
        \item The challenger chooses a bit $b$ needed to set the \textit{Test} response.
        \item Every party and additional initial material is initialized.
        \item $\mathcal{A}$ may issue any queries as described above.
        \item At one point, $\mathcal{A}$ stops and outputs a bit $b^{\prime}$ to indicate its guess for $b$, based on the
response to the \textit{Test} query. The experiment outputs 1 iff $b^{\prime} = b$.
    \end{enumerate}
\end{definition}

\begin{definition} \cite{10.1007/978-3-031-35486-1_24} A key exchange protocol $\Pi$ is $\epsilon$-SK-secure if the following holds for any adversary $\mathcal{A}$:
    \begin{enumerate}
        \item two honest parties (i.e. uncorrupted parties who faithfully execute the protocol instructions) completing matching sessions of the protocol $\Pi$ will output the same key, except with negligible probability.
        \item the advantage of the adversary $\mathcal{A}$ in the key indistinguishability experiment described above is:
        \begin{align*}
            \advantage{Key-Ind}{\Pi}[(\mathcal{A})] = 2\cdot\abs{P[b' = b] - \frac{1}{2}} \leq \epsilon
        \end{align*}
    \end{enumerate}
\end{definition}

\subsection{An alternative commitment-based model}

The framework for AKE models presented above provides a clear picture of the elements needed to model it: first, two environments are required, the AM and the UM. The former represents the authenticated environment, one in which interactions cannot be modified and/or falsified. The latter represents the unauthenticated environment, one in which untrustworthy interactions are to be expected.

In the UM, the model is augmented with an authentication function $I$, which seeks to provide the authentication part of the model. In \cite{cryptoeprint:1998/009}, this function represents an initial phase of authenticated, out-of-bound exchange of cryptographic material. This material is later used to construct the compilers to transform AM-protocols into UM-protocols.

A distinguisher experiment between a protocol in the AM and the UM is also required, which seeks to capture the increase distinguishability capabilities that an adversary gets if the protocol is executed over unauthenticated channels. This definition is employed to define the advantage of an adversary against a protocol in the UM.

In order to transform protocols in the AM to protocols in the UM, the notion of a \textit{compiler} is defined: its role is to create a protocol in the UM from a protocol in the AM such that they are indistinguishable to an adversary. The model then must define ways to construct said compilers from $I$ and its specifics, so they can achieve authentication. 

In the model presented in this work, the above concepts are introduced as follows:

\paragraph*{The Authenticated Model (AM).}
The AM is defined in the exact same terms as the AM from the model above: an adversary which desires to trick parties $P_1, ..., P_n$ is able to delay and/or hold onto messages, but those that get delivered must be done without any alteration to their content, and to the intended recipient. When a message is sent by a party, the message is added to a set M of undelivered messages. When a party receives a message, this message is deleted from $M$, meaning that $P_i$ sent message $m$ to $P_j$ ). 

The attacker has the ability to corrupt any party in the protocol. This means, learning all information pertaining to it. If a party is corrupted, then the attacker is allowed to impersonate them, but such action is recorded accordingly in the party’s protocol output. The global output of the protocol is defined as the concatenation of the output of each party at the end of it, plus the output of the attacker. 

Each party’s output is engrossed by every activation regarding the protocol flow. The adversary’s output is defined to be all the information gathered by the adversary throughout the computation. This includes registration of important events that occurred during the execution, such as corruption of parties.

\paragraph*{The Unauthenticated Model (UM).} The UM in the commitment-based model is very similar to the original UM, but with one key difference: the source of authentication $I$ to be employed. 

The UM is defined as the AM, but with one exception regarding the abilities of an adversary: now it is not restricted to trustful messages (i.e. messages that were in $M$). Instead, it can send messages that were intended to a party $P_j$ to other party $P_k$, or even send messages that are made up.

Under the AKE model presented in this work, the UM is augmented with a \textit{finalization function} $I_f$ which models a final phase of out-of-band and authenticated verification of information between parties under a protocol run. This function models the necessary final phase of authenticated verification of the AV elements that will be generated within each protocol run between any two parties, and that represents the core of the presented AKE model. Formally, we can define $I_f$ as a function
\begin{align*}
    & \{P_1, ..., P_n\} \times \{P_1, ..., P_n\} \times \{0,1\}^n \times \{0, 1\}^n \longrightarrow \{0,1\}\\
    & (P_i, P_j, E, E') \longmapsto \left\{ \begin{array}{lc}
    1 & : \text{$P_j$ authenticates the party $P_i$ $\wedge$ E = E'} \\
    0 & : \text{Otherwise}
    \end{array} \right.
\end{align*}
Alternatively, $I_f$ can be modeled as an AM verification by party $P_j$ of its AV element $E$ against the AV element $E'$ sent by party $P_i$, in the AM.

We also need to define what it means for an adversary to corrupt a party $P_i$, i.e., the capacity to impersonate in an authenticated way the party $P_i$. Therefore, in our model, the corruption of party $P_i$ is defined as the capacity to override the authenticated validation performed by $I_f$ as if it were the corrupted party. The specific way in which this can be practically achieved depends, naturally on the specific of each $I_f$ function employed.

The two environment models, the AM and UM, under the commitment-based AKE model, are related through Definitions \ref{AM-UM-DistExp} and \ref{am-um-emulate}, which are used in the exact same terms under the present model. Likewise, the concept of \textit{compilers}, as defined in Definition \ref{def_compiler}, will also be employed under the presented model.

\subsubsection{Differences between the two models: the authentication phase}
In the model presented in this work, the authentication is not provided by an initial authenticated exchange of cryptography material, but rather of a \textit{final} authenticated verification of certain values that will be generated within each protocol run. Thus, upon finalization of each protocol's message exchange and, as part of its execution, each party will generate an element as specified by the protocol, from elements shared or generated during the protocol execution. 

This value will be verified by the receiver to match the corresponding generated by the initiator, in an authenticated way, by means of a '\textit{verification phase}' modeled by $I_f$. If this verification is satisfactory, then we establish that \textit{m} has been correctly received by party $P_j$ from party $P_i$, and in particular $P_j$ validates, in an authenticated way, $P_i$ as the party with whom they have executed the protocol. 

This means that a message will only be considered as successfully received by a party $P_j$ if it involves a successful validation under $I_f$. The idea behind this process is that it will mean that $m$ was indeed sent by party $P_i$ to party $P_j$, as the possibilities for an adversary to forge this value and get validated between $P_j$ and $P_i$ will be negligible.

This component is essential to the presented model, as seeks to provide the necessary authentication source that is lost once migrating to insecure channels, as the previous UM model did with the initialization function $I$.

While the specifics of how these kind of functions work will be dependent upon each protocol description, the theoretical behavior must always be the same: an authenticated verification of a value generated by each party during the protocol run, in which the responder must corroborate the value generated by the initiator with it's own.

Notice that this verification phase and the values generated to perform the validation will be constructed from elements formed and transmitted during the protocol run. The idea behind is that, if the protocols are constructed in a certain way, any modification of the values exchanged by the adversary will not be validated during $I_f$. 

Note that the definition of $I_f$ implies that authentication is not mutual, i.e. the receiver authenticates the identity of the sender, only. This behavior is also present on the original AKE model, as only the identity of the party holding the private key employed in the protocol is authenticated.

\subsubsection{Constructing compilers in the new model: commitment-based schemes}
We provide the definition of MT-authenticators, as introduced in \cite{cryptoeprint:1998/009}, with regards to the newly introduced UM model:

Consider the message transmission (MT) protocol designed for authenticated networks. The protocol takes the empty input. Upon activation within $P_i$ on external request ($P_j$;\textit{m}), party $P_i$ sends the message ($P_i$;$P_j$;\textit{m}) to party $P_j$, and outputs '$P_i$ sent \textit{m} to $P_j$'. Upon receipt of a message ($P_i$;$P_j$;\textit{m}), $P_j$ outputs '$P_j$ received \textit{m} from $P_i$'.

Then, let $\lambda$ be a protocol that emulates the MT protocol over unauthenticated networks. Define a compiler $C_\lambda$ as follows:

Given a protocol $\pi$, the generated
protocol $\pi'$ = $C_\lambda(\pi)$, running within party $P_j$, first invokes $\lambda$. Next, for each message that $\pi$ sends, $\pi'$ activates $\lambda$ with external request for sending that message to the specified recipient. Whenever
$\pi'$ is activated with some incoming message, it activates $\lambda$ with this incoming message. When $\lambda$ outputs '$P_j$ received \textit{m} from $P_i$', protocol $\pi$ is activated with incoming message \textit{m} from $P_i$.

Under the UM, $\lambda$ outputs '$P_j$ received \textit{m} from $P_i$' when party $P_j$ successfully validates the AV generated on its end from message \textit{m} transmission against the value generated by $P_i$, by means of $I_f$.

The objective of the remainder of this section is to show that, under this new UM model, equivalent results to the ones provided on \cite{cryptoeprint:1998/009} and \cite{10.1007/3-540-44987-6_28} hold. We start by showing that, under a MT-authenticator $\lambda$ as defined above, the compiler $C_\lambda$ constructed from it is an authenticator:
\begin{theorem}
Under the new UM model, let $\lambda$ be an MT-authenticator, and
let $C_\lambda$ be a compiler constructed based on $\lambda$ . Then $C_\lambda$ is an authenticator.
\end{theorem}
\begin{proof}
The proof laid out in \cite{cryptoeprint:1998/009} applies, despite the construction differences with the original UM. That is because despite variations, we are still provided with $\lambda$ being a MT-authenticator. In our model, this will intuitively mean that if an attacker is able to modify (or even fully forge) a message transmission, validation through $I_f$ wont be successful, or equivalently, if the output '$P_i$ received \textit{m} from $P_j$' is generated from $\lambda$, this will necessarily mean that $I_f$ will produce an affirmative validation (and therefore will not happen that the triple ($P_j$,$P_i$,\textit{m}) is not currently in the set $M$ of undelivered messages in the AM, and $P_j$ is uncorrupted).

We cover the most delicate part of the proof: we would need to prove that it is a legitimate behavior to assume that, except with negligible probability, if the activation of the MT-authenticator $\lambda$ generates an output '\textit{party $P_j^{\prime}$ has received message m from $P_i^{\prime}$}', then the element $(P_i, P_j, m)$ is still in the set of undelivered messages, in the AM, when both parties are uncorrupted. And, this is still true as, in our model, the successful reception of a message in the UM involves the affirmative validation of $I_f$, for the uncorrupted sender $P_i^{\prime}$ and, since by hypothesis we are working with a protocol $\lambda$ that emulates MT over unauthenticated network, this situation can only arise with negligible probability, for non-corrupted parties. We refer to \cite{cryptoeprint:1998/009} for the detailed proof.
\end{proof}

The above results shows that, to send a message $m$ under the UM, it is secure to do so under the protocol generated by the compiler $C_{\lambda}$ defined by an MT-authenticator $\lambda$. 

This can be generalized to send a sequence of messages, each under an MT-authenticator $\lambda_i$:
\begin{theorem} \cite{10.1007/978-3-031-35486-1_24} \label{gen_MT_compiler}
Let $\Lambda = (\lambda_1, ... , \lambda_t)$ be a sequence of $t$ MT-authenticators so that each $\lambda_k$ $\epsilon$-emulates MT. Then the compiler, $C_{\Lambda}$, will be an authenticator such that for any protocol $\Pi$ in the AM, $C_{\Lambda}(\Pi)$ ($t \cdot \epsilon$)-emulates $\Pi$ in the UM.
\end{theorem}
\begin{proof}
The proof of \cite{10.1007/978-3-031-35486-1_24} can be applied here as well: Let $\Pi$ be an AM protocol. Let $\mathcal{U}$ be a UM-adversary interacting with $C_{\Lambda}(\Pi)$. The adversary $\adv$
runs $\mathcal{U}$ on a simulated interaction. Action requests from $\mathcal{U}$ to parties in the UM can be
mimicked by $\adv$ in the AM and $\adv$ relays its results back to $\mathcal{U}$. The only problem with the
simulation could occur in the well-known case that $\mathcal{U}$ specifies that \textit{"a message is received by some party $P_j$ from some party $P_i$ in the UM, but that message is not in the set of messages waiting for delivery in the AM"}. But this can happen with probability bounded by $\epsilon$.

Such an event could occur for any of the t messages and so the probability that the simulation is correct is at least $(1 - \epsilon)^t \geq 1 - t\cdot \epsilon$. 
Finally, any observer will be able to
distinguish between the run of $\Pi$ in the AM and $C_{\Lambda}(\Pi)$ in the UM with advantage at most $\epsilon' = t \cdot \epsilon$.
\end{proof}

We define now the MT-authenticator that will serve as the foundation of the secure UM protocols defined in the following sections:
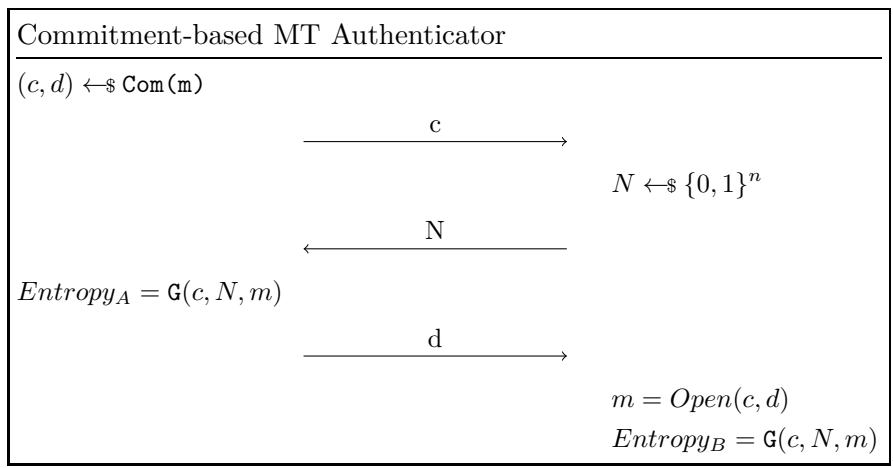
\begin{pchstack}[boxed, center, space=1em]
    \procedure{Commitment-based MT Authenticator}{%
     \textbf{Alice} \> \> \textbf{Bob} \\
     (c, d) \sample \texttt{Com(m)} \> \> \\
     \> \sendmessageright{top=\text{c}} \> \\
     \> \> N \sample \{0,1\}^n \\ 
     \> \sendmessageleft{top=\text{N}} \> \\ 
     E_A = \texttt{G}(B, c, N, m) \> \> \\
     \> \sendmessageright{top=\text{d}} \> \\
     \> \> m^{\prime} = Open(c, d) \\
     \> \> E_A^{\prime} = \texttt{G}(B, c, N, m^{\prime}) \\
	 \> \> \texttt{If $I_f(E_A, E_A^{\prime})$ = 1} \\
	 \> \> \texttt{"$B$ received $m$ from $A$"}
     }
\end{pchstack}
The idea behind the above authenticator is the following:
\begin{itemize}
    \item Alice, who wishes to send a message \textit{m} to Bob, first generates a commitment \textit{c} of such message, and sends it.
    \item Bob, upon reception of such message, just generates a random share \textit{N} and transmits it to Alice. 
    \item When Alice receives this random share \textit{N}, liberates the opening \textit{d} associated with the previous commitment, so Bob can access the message \textit{m}.
    \item Both parties generate a summary of the exchange, by deriving an AV from the commitment \textit{c}, the random share \textit{N}, the actual message \textit{m} and the identity of the intended recipient $B$, which is added to ensure that an attacker cannot redirect legitimate messages to other uncorrupted not-intended destinations. Once the exchange is finished, the responder will validate its AV against the one obtained by Alice, in a secure and authenticated way, as modeled by $I_f$.
\end{itemize}
The security of the above protocol resides in the security properties of commitment schemes and the elements that conform the AV.
We formalize this idea, by proving that the Commitment-based protocol described above emulates MT in unauthenticated networks:

\begin{proposition} \label{CS_MT_security}
The Commitment-based MT-authenticator, when instantiated with a secure Robust Commitment Scheme CS and a random oracle O, emulates MT over unauthenticated channels with advantage 
\begin{multline*}
\epsilon \leq l \cdot (\frac{3}{2^{n_{av}}} + \frac{q_O}{2^{n_{av}}} \cdot ((1 + \advantage{\texttt{Hiding}}{\texttt{CS}} \cdot \advantage{\texttt{CR}}{\texttt{CS}}) \cdot \advantage{\texttt{Hiding}}{\texttt{CS}} \\ 
+ 2 \cdot \advantage{\texttt{Binding}}{\texttt{CS}}) +\advantage{\texttt{Strong-Binding}}{\texttt{CS}})
\end{multline*} 
where $l = n_p^2 \cdot n_m$, $n_p$ is the number of parties running the protocol, $n_m$ the maximum number of message sent by each party, $q_O$ is the maximum number of queries that an adversary issues to $O$ and $n_{av}$ the length of the random oracle $O$.
\end{proposition}
\begin{proof}
We follow the same demonstration ideas of \cite{cryptoeprint:1998/009,10.1007/3-540-44987-6_28,10.1007/978-3-031-35486-1_24}. An AM-adversary can perfectly simulate an UM-adversary unless the event \textit{'In the UM there is the output “Q received m from P” for some parties P and Q, but there was no previous output “P sent m to Q”, for uncorrupted parties P, Q'} happens. We first bound the probability of this event happening on any particular $(P, Q, m)$, and then limit the probability of distinguishing between adversaries by the probability of the above event happening on a particular exchange times the maximum number of exchanges generated.
    
Under the newly defined UM, if the output \textit{"Q received m from P"} happens, in particular it means that $Q$ has trustfully validated $P$ as the party at the other end of the communication. This means that we can establish $P$ as the sender, and $Q$ as the final receiver. It remains to see that, except with negligible probability, it cannot happen that $P$ intended to contact another party $Q^* \neq Q$ or send another message $m^* \neq m$.
    
Therefore, the probability of this event happening on any particular $(P, Q, m)$ is reduced to the probability of an UM adversary to generate different elements, or a different receiver, to those intended in a way that the later verification phase $I_f$ between $P$ and $Q$ is successful. We model that in the game G$_{AV}$ and the advantage against the scheme as
\begin{align*}
    \advantage{}{\texttt{CS-MT}} := \abs{\prob{\texttt{G$_{AV}$} \Rightarrow 1}}
\end{align*}
\begin{pchstack}[ boxed , center, space=1em]
  {\procedure[linenumbering]{Game G$_{AV}$}{\phantomsection\label{CSMT}
      (Q, m) \sample \mathcal{P} \times \mathcal{M} \\
      (c, d) \sample \texttt{Com($m$)}  \\
      (Q^*, c^*, d^*) \sample \adv_1^{O}(c, Q) \\
      N \sample \{0,1\}^n \\
      N^* \sample \adv_2^{O}(c, c^*, d^*, Q, Q^*, N) \\
      v_1 = O(Q, c, N^*, m) \\
      d^{**} \sample \adv_3^{O}(c, c^*, d^*, Q, Q^*, N, N^*, d) \\
      m^{**} = Open(c^*, d^{**}) \\
      v_2 = O(Q^*, c^*, N, m^{**}) \\
      \pcreturn (v_1 = v_2) \wedge \bigwedge (v_i \neq \bot) \wedge (c^*, Q^*, N^*, d^{**}) \neq (c, Q, N, d)
    }}
\end{pchstack}
where $\mathcal{P}, \mathcal{M}$ represent the space of valid \textit{parties} and \textit{messages}, respectively, and $O(P, x_1, x_2, x_3)$ represents the AV calculation that will undergo each party and $P$ is the identity of the intended receiver. Note that this value is added to bind the original receiver intended with the final receiver, and ensures that an adversary cannot simply redirect the original messages to another party. 

It is important to note that, due to Definition \ref{AV_G}, if $(x_1, \cdots, x_n) \neq (x_1^*, \cdots, x_n^*)$, then $TLV(x_1) \conc \cdots \conc TLV(x_n) \neq TLV(x_1^*) \conc \cdots \conc TLV(x_n^*)$. Therefore, for two distinct inputs to generate an output match, a random collision at the oracle $O$ must occur.

We define the intermediate games $G_0$ and $G_1$, where intuitively each game represents an additional layer where the attacker behaves like the AM model: in game $G_0$, the attacker does not modify the flow of the last exchange and in game $G_1$ the attacker only modifies the first exchange.

\begin{pchstack}[ boxed , center, space=1em]
  {\procedure[linenumbering]{Game G$_0$}{\phantomsection\label{G00}
      (Q, m) \sample \mathcal{P} \times \mathcal{M} \\
      (c, d) \sample \texttt{Com($m$)}  \\
      (Q^*, c^*, d^*) \sample \adv_1^{O}(c, Q) \\
      N \sample \{0,1\}^n \\
      N^* \sample \adv_2^{O}(c, c^*, d^*, Q, Q^*, N) \\
      v_1 = O(Q, c, N^*, m) \\
      m^{**} = Open(c^*, d) \\
      v_2 = O(Q^*, c^*, N, m^{**}) \\
      b_1 = (v_1 == v_2) \\
      b_2 = \bigwedge (v_i \neq \bot) \\
      b_3 = (c^*, Q^*, N^*) \neq (c, Q, N) \\
      \pcreturn b_1 \wedge b_2 \wedge b_3
    }}
    
  {\procedure[linenumbering]{Game G$_1$}{\phantomsection\label{G10}
      (Q, m) \sample \mathcal{P} \times \mathcal{M} \\
      (c, d) \sample \texttt{Com($m$)}  \\
      (Q^*, c^*, d^*) \sample \adv_1^{O}(c, Q) \\
      N \sample \{0,1\}^n \\
      v_1 = O(Q, c, N, m) \\
      m^{**} = Open(c^*, d) \\
      v_2 = O(Q^*, c^*, N, m^{**}) \\
      b_1 = (v_1 == v_2) \\
      b_2 = \bigwedge (v_i \neq \bot) \\
      b_3 = (c^*, Q^*) \neq (c, Q) \\
      \pcreturn b_1 \wedge b_2 \wedge b_3
    }}
\end{pchstack}

Therefore, the value $\abs{\prob{\texttt{G$_{AV}$} \Rightarrow 1} - \prob{\texttt{G$_0$} \Rightarrow 1}}$ represents the advantage of the attacker in generating a collision between the oracle output from both parties when being able to modify the opening value received by the responder.

In order to bound this difference, let us define $\mathcal{F}$ as the event that, in Game G$_{AV}$, an attacker has the ability to exploit the additional step $\adv_3$ to win the game. Formally, this means the event the attacker is able to generate a value $d^{**}$ $\neq d$ in step $\adv_3$ such that $Open(c^*, d^{**}) = m^{**}$, satisfying $v_1 = v_2 \neq \bot$. 

Note that, by definition of the games under consideration, since the event $\mathcal{F}$ is the only difference between the games in consideration (i.e. if an attacker is not able to exploit $\adv_3$ to win the game $G_{AV}$, they are equivalent), the application of the difference lemma \cite{cryptoeprint:2004/332} yields
$$\abs{\prob{\texttt{G$_{AV}$} \Rightarrow 1} - \prob{\texttt{G$_0$} \Rightarrow 1}} \leq \prob{\mathcal{F}}$$
Now, it is required to bound the probability of the event $\mathcal{F}$. To do so, two events are considered: \texttt{F$_1$}, in which $(c^*, Q^*, N^*) \neq (c, Q, N)$, and \texttt{F$_2$}, in which $(c^*, Q^*, N^*) = (c, Q, N)$.

Under \texttt{F$_1$}, $(c^*, Q^*, N, m^{**}) \neq (c, Q, N^*, m)$ and therefore the only possible way to achieve $v_1 = v_2$ is via an oracle collision. Therefore, an adversary needs to find $d^{**} \neq d$ such that $Open(c^*, d^{**}) = m^{**}$ and $O(c^*, Q^*, N, m^{**}) = O(c, Q, N^*, m)$.

For this analysis, two more cases are defined: given $m^* = Open(c^*, d^*)$, define \texttt{F$_{1, 1}$} as the event that $O(c^*, Q^*, N, m^*) = O(c, Q, N^*, m)$, and \texttt{F$_{1, 2}$} as the event that $O(c^*, Q^*, N, m^*) \neq O(c, Q, N^*, m)$.

If event \texttt{F$_{1, 1}$} happens, the adversary has already won the game, as a random collision has happened. Therefore, $\prob{\mathcal{F} \cond \texttt{F$_{1, 1}$}, \texttt{F$_{1}$}} = 1$ and $\prob{\texttt{F$_{1, 1}$} \cond \texttt{F$_{1}$}} = \frac{1}{2^{n_{av}}}$.

If \texttt{F$_{1, 2}$} holds, the adversary is required to find $d^{**}$ such that $Open(c^*, d^{**}) = m^{**} \neq m^*$ satisfying $O(c^*, Q^*, N, m^{**}) = O(c, Q, N^*, m)$. Note that the requirement $m^{**} \neq m^*$ is precisely because $\texttt{F$_{1, 2}$}$ ensures that, if $m^{**} = m^*$, then $O(c^*, Q^*, N, m^{**}) = O(c^*, Q^*, N, m^*) \neq O(c, Q, N^*, m)$. 

We claim that the probability $\prob{\mathcal{F} \cond \texttt{F$_{1, 2}$}, \texttt{F$_{1}$}}$ is bounded by the \textit{Combined Collision} advantage: first note that the probability of an adversary to find a $d^{**}$ such that $Open(c^*, d^{**}) = m^{**} \neq m^*$ in the above game is modeled exactly by the following: 
\begin{pchstack}[boxed, center]
  {\procedure[linenumbering]{Game G$_{Binding^*}$}{\phantomsection\label{Gbinding'}
      (c^*, d^*) \sample \adv_0() \\
      d^{**} \sample \adv_1 (c^*, d^*)  \\
      m^* = Open(c^*, d^*)  \\
      m^{**} = Open(c^*, d^{**})  \\
      \pcreturn (m^* \neq m^{**}) \wedge (m^*, m^{**} \neq \bot)
    }}
\end{pchstack}
as the pair $(c^*, d^*)$ was also generated by the adversary at a previous stage. Note that
$$\prob{\texttt{G$_{Binding^*}$} \Rightarrow 1} \leq \prob{\texttt{G$_{Binding}$} \Rightarrow 1}$$ 
(i.e. the game $G_{Binding^*}$ is bounded by the game $G_{Binding}$), as in the later, the adversary is the one to draw the tuple $(c^*, d^*, d^{**})$ at the same time, as opposed to receiving $(c^*, d^*)$ and just generating $d^{**}$. Note that, in this case, the games would actually be equivalent, as the two generation phases correspond to the adversary and not additional information is needed between phases.

Now, we define an adversary $\mathcal{Y}_1$ that returns $m^{**}$ if it can find $d^{**}$ such that $Open(c^*, d^{**}) = m^{**} \neq m^*$, which by the above analysis, has probability of success bounded by $\advantage{\texttt{Binding}}{\texttt{CS}}$. Then, $\prob{\mathcal{F} \cond \texttt{F$_{1, 2}$}, \texttt{F$_{1}$}}$ is exactly $\advantage{\texttt{combined}}{\texttt{CHF}}[(\mathcal{A}, n_{av}, \mathcal{Y}_1)]$ and thus by \ref{combined_adv} $$\prob{\mathcal{F} \cond \texttt{F$_{1, 2}$}, \texttt{F$_{1}$}} \leq q_O \frac{\advantage{\texttt{Binding}}{\texttt{CS}}}{2^{n_{av}}}$$

Under \texttt{F$_2$}, $(c^*, Q^*, N^*) = (c, Q, N)$ and therefore there are only two possible ways to win the game \texttt{F}. The first one, referred to as \texttt{F$_{2, 1}$}, is to generate a collision on the inputs to the oracle function. This means, to find $d^{**} \neq d$ such that $Open(c^*, d^{**}) = Open(c, d^{**}) := m^{**} = m$. 

The probability of such a collision is modeled exactly by the following game
\begin{pchstack}[boxed, center]
  {\procedure[linenumbering]{Game G$_{Strong-Binding^*}$}{\phantomsection\label{GStr_binding'}
  	 m \sample \mathcal{M} \\
      (c, d) \sample \texttt{Com}(m) \\
      d^{**} \sample \adv (c, d)  \\
      m^{**} = Open(c, d^{**})  \\
      \pcreturn (m = m^{**}) \wedge (d^{**} \neq d) \wedge (m \neq \bot)
    }}
\end{pchstack}
and $$\prob{\texttt{G$_{Strong-Binding^*}$} \Rightarrow 1} \leq \prob{\texttt{G$_{Strong-Binding}$} \Rightarrow 1}$$ 
(i.e. the game $G_{Strong-Binding^*}$ is trivially harder that the game $G_{Strong-Binding}$), as in the later, the adversary is the one to draw the tuple $(c, d, d^{**})$, as opposed to receiving $(c, d)$ and just generating $d^{**}$.

Therefore, $$\prob{\texttt{F$_{2, 1}$} \cond \texttt{F$_2$}} = \prob{\texttt{G$_{Strong-Binding^*}$} \Rightarrow 1} \leq \advantage{\texttt{Strong-Binding}}{\texttt{CS}}$$

The second one, referred to as \texttt{F$_{2, 2}$}, is to generate a collision on the outputs, provided that the inputs are not equal. This means, to find $d^{**} \neq d$ such that $Open(c^*, d^{**}) = m^{**} \neq m$ and $O(c^*, Q^*, N, m^{**}) = O(c, Q, N^*, m)$.

In a similar manner to above, we claim that the probability $\prob{ \texttt{F$_{2, 2}$} \cond \texttt{F$_{2}$}}$ is bounded by the \textit{Combined Collision} advantage: first note that the probability of an adversary to find a $d^{**} \neq d$ such that $Open(c, d^{**}) = m^{**} \neq m$ is modeled exactly by the Game G$_{Binding^*}$.

Now, we define an adversary $\mathcal{Y}_2$ that returns $m^{**}$ if it can find $d^{**} \neq d$ with the above conditions. Then, $\prob{\texttt{F$_{2, 2}$} \cond \texttt{F$_{2}$}}$ is exactly $\advantage{\texttt{combined}}{\texttt{CHF}}[(\mathcal{A}, n_{av}, \mathcal{Y}_2)]$ and thus by \ref{combined_adv} 
$$\prob{ \texttt{F$_{2, 2}$} \cond \texttt{F$_{2}$}} \leq \prob{\texttt{G$_{Binding^*}$} \Rightarrow 1} \cdot \frac{q_O}{2^{n_{av}}} \leq \advantage{\texttt{Binding}}{\texttt{CS}} \cdot \frac{q_O}{2^{n_{av}}}$$

Together, we have 
\begin{align*}
& \prob{\mathcal{F}} = \prob{\mathcal{F} \cond \texttt{F$_1$}} \cdot \prob{\texttt{F$_1$}} + \prob{\mathcal{F} \cond \texttt{F$_2$}} \cdot \prob{\texttt{F$_2$}} \leq \\
& \prob{\mathcal{F} \cond \texttt{F$_{1, 1}$}, \texttt{F$_1$}} \cdot \prob{\texttt{F$_{1, 1}$} \cond \texttt{F$_1$}} + \prob{\mathcal{F} \cond \texttt{F$_{1, 2}$}, \texttt{F$_1$}} \cdot \prob{\texttt{F$_{1, 2}$} \cond \texttt{F$_1$}} + \prob{\mathcal{F} \cond \texttt{F$_2$}} \leq \\
& \frac{1}{2^{n_{av}}}\cdot (1 + q_O \cdot \advantage{\texttt{Binding}}{\texttt{CS}}) + \prob{\texttt{F$_{2, 1}$} \cond \texttt{F$_2$}} + \prob{\texttt{F$_{2, 2}$} \cond \texttt{F$_2$}} \leq \\
& \frac{1}{2^{n_{av}}}\cdot (1 + 2 \cdot q_O \cdot \advantage{\texttt{Binding}}{\texttt{CS}}) +\advantage{\texttt{Strong-Binding}}{\texttt{CS}}
\end{align*}

The value $\abs{\prob{\texttt{G$_0$} \Rightarrow 1} - \prob{\texttt{G$_1$} \Rightarrow 1}}$ represents the advantage of the attacker in generating a collision between the oracle output from both parties when being able to modify the random challenge value sent by the responder.

In order to bound this difference, let us define $\mathcal{H}$ as the event that in game \texttt{G$_0$} an adversary has the ability to exploit the second exchange to win the game, i.e., to generate $N^* \neq N$ in step $\adv_2$ such that $v_1 = v_2 \neq \bot$. Note that, by definition of the games under consideration, the application of the difference lemma \cite{cryptoeprint:2004/332} yields
$$\abs{\prob{\texttt{G$_{0}$} \Rightarrow 1} - \prob{\texttt{G$_1$} \Rightarrow 1}} \leq \prob{\mathcal{H}}$$ 
Now, it is required to bound the probability of the above event. There are exactly two ways: via a random generation of $N^*$ and collision between outputs (referred to as $H_1$), and by trying to extract information and generating a targeted value for $N^*$ (i.e. not trying random collision), referred to as $H_2$.

The advantage in winning the game from the first option is bounded by the probability of a random collision, which is $\frac{1}{2^{n_{av}}}$. Therefore, we will focus on the probabilities for an adversary in making an elaborated guess.

The analysis is done in terms of the event $D$ that the adversary knows the value of $m$. First, we will analyze the probability of $D$ itself. Since $m = Open(c, d)$, the pair $(c, d)$ represents a valid opening of $m$. Therefore, this probability is bounded by the advantage in retrieving the committed input from a valid commitment. This is modeled by the following game:
\begin{pchstack}[boxed, center]
  {\procedure[linenumbering]{Game G$_{Hiding^*}$}{\phantomsection\label{G_Hiding'}
      m \sample \mathcal{M} \\
      (c, d) \sample Com(m) \\
      m' \sample \adv(c) \\
      \pcreturn (m = m')
    }}
\end{pchstack}

And, in turn, this advantage is bounded by the ability to distinguish between two committed inputs from a commitment value, in a traditional $IND \implies OW$ argument. The latter represents the Hiding advantage of a commitment scheme, and therefore it holds that $\prob{D} \leq \advantage{\texttt{Hiding}}{\texttt{CS}}$.

In case $D$ does not hold, the adversary does not possess all the information required to make an informed guess: the value $v_1$ is dependent on both the choice of $N^*$ of the adversary and $m$, which is defined but not known to the adversary. Therefore, $\prob{\texttt{H$_2$} \cond \neg D} = 0$, as the only possible way when not all information is known is random collision.


In case $D$ does hold, we analyze the success probability in terms of the event $C$ that $c^* = c$. If $C$ is also true, then $Open(c^*, d) = m^{**} = m = Open(c, d)$. Therefore, since $m$ is known, $m^{**}$ is also known. This provides the adversary with the knowledge of all elements involved in the generation of both $v_1$ and $v_2$. 

Thus the success probability $\prob{\texttt{H$_2$} \cond C,D}$ resides in the adversaries ability to generate a random element $N^* \neq N$ such that $O(Q^*,c, N, m) = O(Q, c, N^*, m)$. Since the attacker cannot win the game via input collision, as it must hold that $N^* \neq N$, the only viable option is to search for a preimage-collision over random elements $N^* \in \{0,1\}^n$, which has probability at most $\frac{q_O}{2^{n_{av}}}$, with $q_O$ being the maximum number of queries made to the oracle $O$.

In case $D$ holds but $C$ does not, further information cannot be claimed about $m^{**}$. For any adversary to be able to win the game without random collision, three conditions must happen simultaneously: first, $m^{**} \neq \bot$. Second, to be able to retrieve $m^{**}$ from $c^*$. Lastly, the ability to force a collision via generating $N^*$.

For the analysis of first condition, since $m^{**} = Open(c^*, d)$ and $c^* \neq c$, it is required to bound the probability that $m^{**} \neq \bot$. It can be bounded by the advantage of the CR property of commitment schemes. Note, however, that this bound is very conservative, since the CR-CS game allows the adversary to find $(c, c', d)$, while in here they are all established.

The second condition requires to bound the probability of retrieving its value. Since $m^{**} = Open(c^*, d) \neq \bot$, the pair $(c^*, d)$ form a valid commitment pair of the value $m^{**}$. Therefore, this probability is bounded by the advantage in retrieving the committed input from a valid commitment. As above, this can be bounded by the Hiding property of the CS.

The third condition requires to bound the probability of being able search for a preimage-collision over random elements $N^* \in \{0,1\}^n$. It has probability at most $\frac{q_O}{2^{n_{av}}}$, with $q_O$ being the maximum number of queries made to the oracle $O$.

Together, this means that $$\prob{\texttt{H$_2$} \cond \neg C,D} \leq \advantage{\texttt{Hiding}}{\texttt{CS}} \cdot \advantage{\texttt{CR}}{\texttt{CS}} \cdot \frac{q_O}{2^{n_{av}}}$$
and therefore
\begin{align*}
& \prob{\mathcal{H}} \leq \prob{\texttt{H$_1$}} + \prob{\texttt{H$_2$}} \leq \\
& \frac{1}{2^{n_{av}}} + \prob{\texttt{H$_2$} \cond \neg D} \cdot \prob{\neg D} + \prob{\texttt{H$_2$} \cond D} \cdot \prob{D} = \\
& \frac{1}{2^{n_{av}}} + (\prob{\texttt{H$_2$} \cond C, D} \cdot \prob{C \cond D} + \prob{\texttt{H$_2$} \cond \neg C, D} \cdot \prob{\neg C \cond D}) \cdot \prob{D} \leq \\
& \frac{1}{2^{n_{av}}} + (1 + \advantage{\texttt{Hiding}}{\texttt{CS}} \cdot \advantage{\texttt{CR}}{\texttt{CS}}) \cdot \frac{q_O}{2^{n_{av}}} \cdot \advantage{\texttt{Hiding}}{\texttt{CS}}
\end{align*}

Lastly, the advantage of the game $G_{1}$ is exactly the probability of a random collision based on the modifications $(Q^*, c^*, d^*)$ done by the adversary. 

It is straightforward to see that, regardless of the information that can be extracted from $c$ or generated from it, a random collision on the oracle is still required, since $N$ has not yet been generated and $(c^*, Q^*) \neq (c, Q)$. The probability of this collision is exactly $\frac{1}{2^{n_{av}}}$.

Together, we have the following:
\begin{align*}
    & \abs{\prob{\texttt{G$_{AV}$} \Rightarrow 1}} \leq \\ 
    & \abs{\prob{\texttt{G$_{AV}$} \Rightarrow 1} - \prob{\texttt{G$_0$} \Rightarrow 1}} + \abs{\prob{\texttt{G$_{0}$} \Rightarrow 1} - \prob{\texttt{G$_1$} \Rightarrow 1}} + \abs{\prob{\texttt{G$_1$} \Rightarrow 1}} \leq \\
    & \frac{3}{2^{n_{av}}} + \frac{q_O}{2^{n_{av}}} \cdot ((1 + \advantage{\texttt{Hiding}}{\texttt{CS}} \cdot \advantage{\texttt{CR}}{\texttt{CS}}) \cdot \advantage{\texttt{Hiding}}{\texttt{CS}} \\
    & + 2 \cdot \advantage{\texttt{Binding}}{\texttt{CS}}) +\advantage{\texttt{Strong-Binding}}{\texttt{CS}}
\end{align*}
    Then, considering all possible messages between parties, which equals the maximum number of messages sent by each party times the square of the number of parties, we arrive to the desired bound.
\end{proof}

\subsubsection{SK security under the new model}
Under our model, we adopt the definition and the setting for key exchange protocols and \textit{session-key security} introduced in \cite{10.1007/3-540-44987-6_28} and detailed in Section \ref{or_sk_security}, but introducing the authentication differences pertaining to our model.

Since the key exchange procedures are a special case of message-driven protocols, the notion of AV and posterior verification are translated onto these protocols, under the UM. An AV (pertaining to specific elements, depending upon the protocol in consideration) is generated by each party running the key exchange protocol. 

As part of the protocol's message exchange, the protocol will trigger each party to verify its AV(s) against the one(s) generated by the opposing party, in a secure and authenticated way, as modeled by $I_f$. This will be represented in the definition of the \textit{Send} query of an adversary against the \textit{Key-Ind} game.

Note that the presence of more than one AV will depend on the particularities of how each protocol is constructed. For example, if the protocol in the UM is built from repeated applications of the CS-based compiler above, then one AV will appear for each message to be exchanged under the protocol in the AM. This means that both the initiator and the responder can trigger the verification function $I_f$ to perform an AV verification, under a single protocol run.

If instead the key exchange is constructed directly, the number of AV values to be verified will depend on the key exchange's specifications.

The definition of SK security employed is as introduced in Section \ref{or_sk_security}, with three modifications: the definition of \textit{completeness} of a session in the UM, the definition of the \textit{Send} query, and the definition of the \textit{Corrupt} ability of the adversary.

The session will be defined to be \textit{completed} if $I_f$ verification of the AV elements generated in the protocol session is affirmative and the key value $\kappa$ is not null (that is, a previous abort did not happen). When a session is completed, the secret value $\kappa$ is appended as local output to both parties. If this verification is not affirmative, we consider the session aborted and $\kappa$ is set to null, even in the case when a non-null secret value of the key was reached at the end of the protocol's message exchange.

The definition of the \textit{Send} query needs to be expanded to include the capacity that a party, upon reception of certain messages in a protocol run, triggers the initiation of the verification phase $I_f$, if no prior rejection happens. This will usually happen at the end of the protocol run by the responder, but could also happen at other stages if more that one AV value needs to be verified within a single session.

The definition of the \textit{Corrupt} capability for the new model is as introduced above. I.e., is defined as the capacity to override the authenticated validation performed by $I_f$ as if it were the corrupted party.

With these modifications, we state the SK security under MT-authenticators theorem, which still holds true under the commitment-based model:
\begin{theorem}
\label{ThCompilerSK}
Let $\pi$ be an $\epsilon$-SK-secure protocol in the AM and let $C_\Lambda$ be a compiler based on a sequence of MT-authenticators $\Lambda$. Consider $\pi' = C_\Lambda(\pi)$ the protocol that $\alpha$-emulates $\pi$ in our modified UM model. Then, $\pi'$ is an $\epsilon'$-SK-secure protocol in our modified UM model, with $\epsilon' = \epsilon + \alpha$.
\end{theorem}
\begin{proof}
The proof presented in \cite{10.1007/978-3-031-35486-1_24} to prove this result under the original model still applies to our model. The definitions of \cite{10.1007/978-3-031-35486-1_24} still models the notion of the key indistinguishability experiment (without any long-term key) and SK-security under our model, with three modifications: 

First, we have included the necessary AV validation process $I_f$ of our general message-driven model inside the definition of \textit{completeness} of a key exchange protocol execution in the UM (i.e., the key $\kappa$ will only be non-null if the exchange has been successful, and that includes the $I_f$ phase was successfully verified). 

Then, the \textit{Send} capability of an adversary has been expanded to include the capacity of the responder to trigger the verification phase $I_f$, when a protocol run is completed.

Lastly, the \textit{Corrupt} capability of an adversary has been re-defined to atone for identity usurpation under the new UM, i.e., corruption of the $I_f$ process, which represents the authentication phase equivalent to the presence of long-term keys in the original UM.
    
We note that these requirements are the appropriate execution over key exchange protocols, in the UM. The key event '\textit{forge}' that models the proof is the well-known event that a party $P_j$ successfully receives a message \textit{m} from another party $P_i$, but $P_i$ never sent it. And the probabilities modeled in \cite{10.1007/978-3-031-35486-1_24} remain the same, under the new UM model.
\end{proof}
\section{KA-based protocols}

Cryptographic protocols based on key exchange algorithms are vastly used in a number of different scenarios, a great deal of those through insecure channels. These protocols are defined in terms of the number of exchanged messages needed to complete them, and the information sent on each message. We provide in this section with a selected overview of how these KA-based protocols would look like, under our new models.

\subsection{A KA-based SK-secure protocol in the AM: Choice of primitive}

The minimum requirement for a KA-based shared secret establishment is 2 messages. The protocol would go as follows:

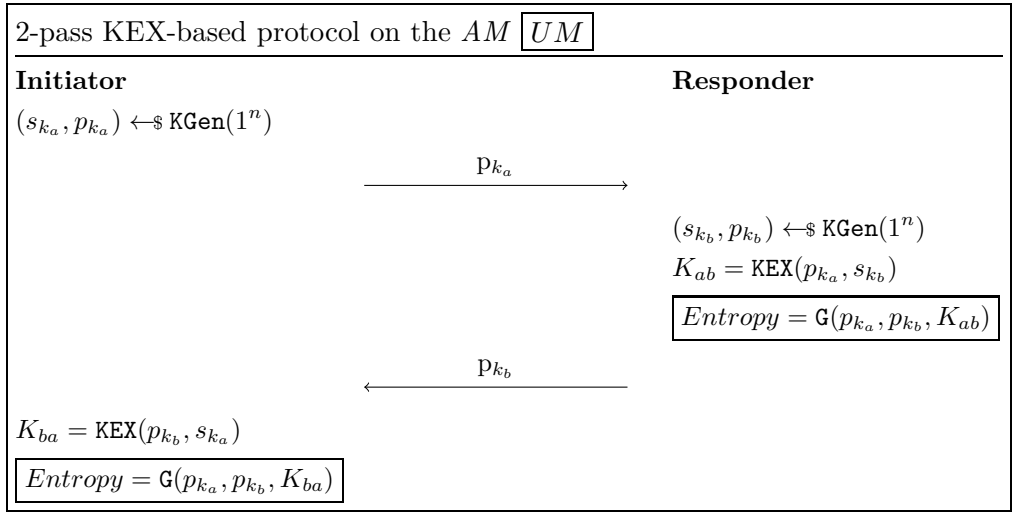
\begin{pchstack}[boxed, center, space=1em]
    \procedure{2-pass KA-based protocol on the \textit{AM} \pcbox{UM}}{%
     \textbf{Alice} \> \> \textbf{Bob} \\
     (s_{k_a}, p_{k_a}) \sample \texttt{KGen} (\secparam) \> \> \\
     \> \sendmessageright{top=\text{(p$_{k_a}$, s)}} \> \\
     \> \> (s_{k_b}, p_{k_b}) \sample \texttt{KGen} (\secparam) \\ 
     \> \> K_{ab} = \texttt{KA}(p_{k_a}, s_{k_b}) \\ 
     \> \> \pcbox{E = \texttt{G}(B, p_{k_a}, s, p_{k_b}, K_{ab})} \\ 
     \> \sendmessageleft{top=\text{(p$_{k_b}$, s)}} \> \\ 
     K_{ba} = \texttt{KA}(p_{k_b}, s_{k_a}) \> \> \\
     \pcbox{E = \texttt{G}(B, p_{k_a}, s, p_{k_b}, K_{ba})} \> \>}
\end{pchstack}

\begin{itemize}
    \item Alice generates a cryptographic key pair ($s_{k_a}$, $p_{k_a}$), by means of the KGen() function of the specific KA selected. Then, the public value $p_{k_a}$ is sent to Bob, along with a session identifier $s$.
    \item  Bob, upon reception of Alice's public key $p_{k_a}$, generates another key pair ($s_{k_b}$, $p_{k_b}$) with the same algorithm KGen(), to then execute the KeyAgreement function, with inputs its own secret key $s_{k_b}$ and the other end's public key $p_{k_a}$. The execution of this function yields a secret key $K_{ab}$. Then, Bob sends its own public key $p_{k_b}$ to Alice.
    \item With Bob's public key, Alice executes the same KeyAgreement function, on its own private key and the received public key, to derive the shared secret $K_{ba}$. If no interference has happened, the KA algorithm ensures that $K_{ab}$ = $K_{ba}$. On the UM, as part of the protocol, the AV elements and the verification phase would appear.
\end{itemize}

\subsubsection{Security \label{3.1.1}}

The security of the protocol itself, without any interference, i.e., on the authenticated model, is based upon the security of the underlying key exchange algorithm selected, and the unique possession by both ends of their respective secret keys. An example is shown on Section \ref{2pass_dh_am}.

Over insecure channels, i.e., on the unauthenticated model, neither party has any assurance that the public key values received indeed corresponds to the actual values sent by the other party. This particular situation is the very foundation of the presence of the session AV calculated by both ends on the protocol. But, even the presence of such value is not enough to avoid MitM interference. The attack against this protocol goes as follows:
\begin{pchstack}[boxed, center, space=1em]
    \procedure[colspace=-1.75cm]{MitM attack on 2-pass KA-based protocol}{%
     \textbf{Alice}  \> \> \textbf{Mallory} \> \> \textbf{Bob} \\
     (s_{k_a}, p_{k_a}) \sample \texttt{KGen} (\secparam) \> \> \> \> \\
     \> \sendmessageright{top=\text{(p$_{k_a}$, s)}} \> \> \> \\
     \> \> (s_{k_{eb}}, p_{k_{eb}}) \sample \texttt{KGen} (\secparam) \> \> \\ 
     \> \> \> \sendmessageright{top=\text{(p$_{k_{eb}}$, s)}} \> \\
     \> \> \> \> (s_{k_b}, p_{k_b}) \sample \texttt{KGen} (\secparam) \\ 
     \> \> \> \> K_{be} \sample \texttt{KA}(p_{k_{eb}}, s_{k_b}) \\
     \> \> \> \>  E_{be} = \texttt{G}(B, p_{k_{eb}},  s, p_{k_b}, K_{be}) \\
     \> \> \> \sendmessageleft{top=\text{(p$_{k_b}$, s)}} \> \\
     \> \> K_{eb} = \texttt{KA}(p_{k_b}, s_{k_{eb}}) \> \> \\
     \> \> E_{eb} = \texttt{G}(B, p_{k_{eb}}, s, p_{k_b}, K_{eb}) \> \> \\
     \> \> \texttt{While $E_{ea} \neq E_{eb}$ do:} \> \> \\
     \> \> \quad (s_{k_{ea}}, p_{k_{ea}}) \sample \texttt{KGen} (\secparam) \> \> \\
     \> \> \quad K_{ea} \sample \texttt{KA}(p_{k_a}, s_{k_{ea}}) \> \> \\
     \> \> \quad E_{ea} = \texttt{G}(B, p_{k_a}, s, p_{k_{ea}}, K_{ea}) \> \> \\
     \> \sendmessageleft{top=\text{(p$_{k_{ea}}$, s)}} \> \> \> \\
     K_{ae} = \texttt{KA}(p_{k_{ea}}, s_{k_a}) \> \> \> \> \\
     E_{ae} = \texttt{G}(B, p_{k_a}, s, p_{k_{ea}}, K_{ae}) \> \> \> \>}
\end{pchstack}

\begin{itemize}
    \item When Alice sends its public key $p_{k_{a}}$, Mallory intercepts it and instead transmits the public key $p_{k_{eb}}$ of its key pair generated. Then, Bob will execute the corresponding KeyAgreement function over the fraudulent public key $p_{k_{eb}}$, generating a secret $K_{eb}$, and its session AV, based on the secret, its own public key $p_{k_{b}}$ and the fraudulent public key value $p_{k_{eb}}$. Then, Bob transmits $p_{k_{b}}$.
    \item Mallory would then need to intercept this public value and substitute with the public key value $p_{k_{ea}}$ of a newly generated key pair ($s_{k_{ea}}$, $p_{k_{ea}}$). But this value must be such that the session AV generated by the shared secret $K_{ae}$ and the public values $p_{k_{ea}}$ and $p_{k_a}$ is the same as the session AV generated by Bob (to which Mallory, as has actively modified the natural course of the protocol, has access to). Therefore, Mallory would need to loop through key pairs ($s_{k_{ea_i}}$, $p_{k_{ea_i}}$) until one satisfies the required condition.
\end{itemize}
The actual success probability of this attack is closely tied with the size of the session AV generated, but is nevertheless a plausible loophole for an attacker to exploit.

In fact, the above diagram provides the following lemma:
\begin{proposition}
The 2-pass KA-based protocol has SK-security in the UM model, regardless of the KA scheme employed, bounded by the following value:
\begin{align*}
    \advantage{\texttt{Key-Ind}}{\pi_{UM}}[(\adv)] \geq 1 - \left(1 - \frac{1}{2^l}\right)^{\min\{q, |Y|\}}
\end{align*}
where $l$ is the output length of the AV random oracle, $q$ is the maximum number of queries the adversary is able to make against the AV random oracle and $Y: = \{(\pk, K) : (\sk, \pk) \sample \texttt{KGen}(), K \sample \texttt{KA}(\pk_p, \sk)\}$ is the domain of all possible combinations of valid public keys and shared secrets with public key $\pk_p$ of the KA scheme.
\end{proposition}
\begin{proof}
The above diagram can be replicated exactly as a routine for an adversary $\adv$ in the UM unless it is unable to generate an AV collision between $E_{ab}$ and $E_{eb}$.

Note the advantage of an adversary in generating said collision is defined exactly as $\advantage{\texttt{combined}}{\texttt{CHF}}[(\adv, l, \mathcal{Y})]$, where $\mathcal{Y}$ is an algorithm that returns values from the domain $Y$ by first generating an ephemeral key pair, and then executing the KA procedure with $\pk_p$. 

Therefore, since $Y$ does not impose any cryptographic restriction that the pairs $(\pk, K)$ must verify, we have $\advantage{\texttt{sample}}{Y}[(\mathcal{A})] = 1$. Consequently, the adversary can construct a set $Y_q$ of $\min\{q, |Y|\}$ different values $((\pk_i, K_i))$.

Since the $SK$ advantage can be lower bounded by the advantage of any plausible attack, we have that:
\begin{align*}
    \advantage{\texttt{Key-Ind}}{\pi_{UM}}[(\adv)] \geq \advantage{\texttt{combined}}{\texttt{CHF}}[(\adv, l, \mathcal{Y})] = 1 - \left(1 - \frac{1}{2^l}\right)^{\min\{q, |Y|\}}
\end{align*}
where the last equality is provided by Proposition \ref{combined_adv_with_Y}.
\end{proof}

While we do not claim that this protocol is always insecure, it is important to realize that, depending upon the capabilities of the adversary in terms of the maximum number of queries $q$, the protocol could become insecure, especially given the fact that, as noted above, the value $l = n_{av}$ will not be high, in general. 

This realization motivates the need to come up with more robust protocols in which the security is guaranteed.

\subsubsection{Practical instantiation of a KA-based exchange \label{2pass_dh_am}}
We provide a practical example of the above protocol, by instantiating the general KA scheme with a basic Diffie-Hellman procedure, over a generic finite cyclic group:

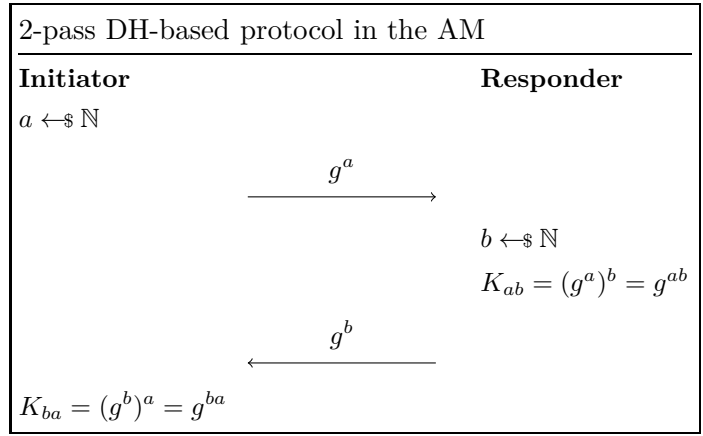
\begin{pchstack}[boxed, center, space=1em]
    \procedure{2-pass DH-based protocol in the AM}{%
     \textbf{Alice} \> \> \textbf{Bob} \\
     a \sample \mathbb{N} \> \> \\
     \> \sendmessageright*[2.5cm]{\text{($g^a$, s)}} \> \\
     \> \> b \sample \mathbb{N} \\ 
     \> \> K_{ab} = (g^a)^b = g^{ab} \\ 
     \> \sendmessageleft*[2.5cm]{\text{($g^b$, s)}} \> \\ 
     K_{ba} = (g^b)^a = g^{ba} \> \>}
\end{pchstack}

This protocol was proven secure on the AM model in \cite{cryptoeprint:1998/009,10.1007/3-540-44987-6_28}, if the DDH assumption holds for the selected group:
\begin{proposition}\cite{10.1007/3-540-44987-6_28}\label{sk_security_2_pass_dh}
    Let $\adv$ be an adversary against the SK-security of the above protocol, which interacts with at most $l$ sessions. Then the 2-pass DH protocol is $\epsilon$-SK-secure in the AM, with 
    $$ \epsilon \leq l \cdot \advantage{DDH}{DH}$$
\end{proposition}

\subsection{A KA-based SK-secure protocol in the UM: Compiler application}
The above protocol is shown to be secure, for practical instantiations, on the AM. To avoid the vulnerabilities that the protocol presented, the application of a compiler based on the MT-authenticator defined on section can be applied, yielding the following (unoptimized) protocol:

\begin{pchstack}[boxed, center, space=1em]
    \procedure{6-pass KA-based protocol (unoptimized)}{%
     \textbf{Alice} \> \> \textbf{Bob} \\
     (\pk_a, \sk_a) \sample \texttt{KGen} (\secparam) \> \> \\
     (c(\pk_a), d(\pk_a)) \sample \texttt{Com}(\pk_a) \> \> \\
     \> \sendmessageright{top=\texttt{(c($\pk_a$), s)}} \> \\
     \> \>  N_B \sample \texttt{$\bin^n$} \\
     \> \sendmessageleft{top=\texttt{$N_B$}} \> \\
     AV_A = \texttt{G}(B, s, \pk_a, c(\pk_a), N_B) \> \> \\
     \> \sendmessageright{top=\texttt{d($\pk_a$)}} \> \\
     \> \> \pk_a = Open(c(\pk_a), d(\pk_a)) \\
     \> \> AV_A = \texttt{G}(B, s, \pk_a, c(\pk_a), N_B) \> \> \pclb
     \pcintertext[dotted]{$\lambda_I - \lambda_R$ division} \\
     \> \> (\pk_b, \sk_b) \sample \texttt{KGen}(\secparam) \\ 
     \> \> (c(\pk_b), d(\pk_b)) \sample \texttt{Com}(\pk_b) \\
     \> \> K = \texttt{KA}(\sk_b, \pk_a) \\
     \> \sendmessageleft{top=\texttt{(c($\pk_b$), s)}} \> \\ 
      N_A \sample \texttt{\{0,1\}$^n$} \> \> \\
     \> \sendmessageright{top=\texttt{$N_A$}} \> \\ 
     \> \> AV_B = \texttt{G}(A, s, \pk_b, c(\pk_b), N_A) \> \> \\ 
     \> \sendmessageleft{top=\texttt{d($\pk_b$)}} \> \\ 
     \pk_b = Open(c(\pk_b), d(\pk_b)) \> \> \\
     K = \texttt{KA}(\sk_a, \pk_b) \> \> \\
     AV_B = \texttt{G}(A, s, \pk_b, c(\pk_b), N_A)}
\end{pchstack}
where each of the two separated blocks represent an almost canonical application of the CS-based MT-authenticator to one message of the 2-pass KA-based protocol.

The only conceptual difference is that the MT-authenticators provide an additional value in the first message exchanged, i.e. the session identifier, which is also appended to the AV calculation.

It is straightforward to see that the defined MT-authenticator accepts any number of inputs to the first message without modifying the security emulation bound provided in Proposition \ref{CS_MT_security}, as long those elements are included in the AV calculation and do not provide any information regarding other values employed in the MT-authenticator. In other words, as long as they are independent of both the commitment generated and the random challenge.

\subsubsection{Security}
The security of the above protocol is a direct application of propositions \ref{CS_MT_security}, \ref{gen_MT_compiler} and \ref{ThCompilerSK}, along with the security of the 2-pass KA-based protocol in the AM. Formally:

\begin{proposition}\label{gen_4_pass_KA_sk_security}
    Let the 2-pass KA-based protocol be a $\epsilon$-SK-secure protocol in the AM. Then, the 6-pass KA-based protocol is $(\epsilon + 2\cdot \alpha)$-SK secure in the UM,
    where $\alpha$ is the MT emulation margin proven in Proposition \ref{CS_MT_security}.
\end{proposition}
\begin{proof}
    The CS-based MT-authenticator, denoted as $\lambda$,  $\alpha$-emulates the MT-authenticator by Proposition \ref{CS_MT_security}.

    Then, Proposition \ref{gen_MT_compiler} proves that the sequence $\Lambda = (\lambda, \lambda)$ defines a compiler $C_{\Lambda}$ such that, for any protocol $\Pi$ in the AM, $C_{\Lambda}(\Pi)$ $2\cdot \alpha$ emulates $\Pi$ in the UM.

    Then, by Theorem $\ref{ThCompilerSK}$, $C_{\Lambda}(\Pi)$ is an $\epsilon'$-SK-secure protocol in the UM, where $\epsilon' := \epsilon + 2 \cdot\alpha$.
\end{proof}

\subsubsection{Protocol optimization}
The above protocol represents an application of the CS-based compiler to the 2-pass KA-based key exchange. Since each application of the MT-authenticator \ref{CSMT} generates three messages, this yields a protocol of $3 \cdot 2 = 6$ messages, whose security, as shown above, is a direct application of the results of the above section.

However, a simple realization shows that, since each application of the MT-authenticator is independent, some messages can be sent \textit{in parallel}, thus reducing the number of exchanged messages.

The optimized version of the above protocol is as follows:
\begin{pchstack}[boxed, center, space=1em]
    \procedure{4-pass KA-based protocol (optimized)}{%
     \textbf{Alice} \> \> \textbf{Bob} \\
     (\pk_a,\sk_a) \sample \texttt{KGen} (\secparam) \> \> \\
    (c(\pk_a), d(\pk_a)) \sample \texttt{Com}(\pk_a) \> \> \\
     \> \sendmessageright*[2.5cm]{\texttt{(c($\pk_a$), s)}} \> \\
     \> \>  N_B \sample \texttt{\{0,1\}$^n$} \\
     \> \> (\pk_b,\sk_b) \sample \texttt{KGen} (\secparam) \\
     \> \> (c(\pk_b), d(\pk_b)) \sample \texttt{Com}(\pk_b) \\
     \> \sendmessageleft*[2.5cm]{\texttt{(c($\pk_b$), s, $N_B$)}} \> \\ 
      N_A \sample \texttt{\{0,1\}$^n$} \> \> \\
     AV_A = \texttt{G}(B, s, \pk_a, c(\pk_a), N_B) \> \> \\ 
     \> \sendmessageright*[2.5cm]{\texttt{(d($\pk_a$), $N_A$)}} \> \\
     \> \> \pk_a = Open(c(\pk_a), d(\pk_a)) \\
     \> \> K_{ba} = KA(\pk_a, \sk_b)) \\
     \> \> AV_A^{\prime} = \texttt{G}(B, s, \pk_a, c(\pk_a), N_B) \> \> \\
     \> \> AV_B = \texttt{G}(A, s, \pk_b, c(\pk_b), N_A) \> \> \\ 
     \> \sendmessageleft*[2.5cm]{\texttt{d($\pk_b$)}} \> \\ 
     \pk_b = Open(c(\pk_b), d(\pk_b)) \> \> \\
     K_{ab} = KA(\sk_a, \pk_b) \> \> \\
     AV_B^{\prime} = \texttt{G}(A, s, \pk_b, c(\pk_b), N_A)}
\end{pchstack}
\begin{itemize}
    \item First, Alice generates a key pair $(\sk_a, \pk_a)$ and sends the commitment $c(\pk_a)$ of the public key value $\pk$, along with the session identifier $s$.
    \item Bob, upon reception of the commitment $c(\pk_a)$, generates its own key pair $(\sk_b, \pk_b)$. Then, generates a commitment \textit{c($\pk_b$)} of the public ket and sends it, along with a random challenge $N_B$.
    \item The inititator, upon reception of the commitment value \textit{c($\pk_b$)}, generates its own challenge $N_A$ and sends it along with the opening value \textit{d($\pk_a$)} of $pk_a$.
    \item Bob verifies the commitment received at the beginning with the opening value received in this iteration and the actual value committed and, upon successful verification, calculates the session AVs, along with the session key and sends the opening value \textit{d($\pk_b$)} of his public key $\pk_b$ committed before.
    \item Alice verifies that the value committed matches the commitment stored, and then proceeds to derive the shared secret key using the \textit{KA} function, and her corresponding session AVs.
\end{itemize}

\subsubsection{Practical instantiation \label{3.2.2}}
We give a practical example of the above 4-pass protocol by building on the example laid out in \ref{2pass_dh_am}:

\begin{pchstack}[boxed, center, space=1em]
    \procedure{4-pass KA-based protocol (optimized)}{%
     \textbf{Alice} \> \> \textbf{Bob} \\
     a \sample \mathbb{N} (\secparam) \> \> \\
     (c(g^a), d(g^a)) \sample \texttt{Com}(g^a) \> \> \\
     \> \sendmessageright*[2.5cm]{\texttt{(c($g^a$), s)}} \> \\
     \> \>  N_B \sample \texttt{\{0,1\}$^n$} \\
     \> \> b \sample \mathbb{N} \\
     \> \> (c(g^b), d(g^b)) \sample \texttt{Com}(g^b) \\
     \> \sendmessageleft*[2.5cm]{\texttt{(c($g^b$), s, $N_B$)}} \> \\ 
      N_A \sample \texttt{\{0,1\}$^n$} \> \> \\
     AV_A = \texttt{G}(B, s, g^a, c(g^a), N_B) \> \> \\ 
     \> \sendmessageright*[2.5cm]{\texttt{(d($g^a$), $N_A$)}} \> \\
     \> \> g^a = Open(c(g^a), d(g^a)) \\
     \> \> g^{ba} = (g^a)^b \\
     \> \> AV_A^{\prime} = \texttt{G}(B, s, g^a, c(g^a), N_B) \> \> \\
     \> \> AV_B = \texttt{G}(A, s, g^b, c(g^b), N_A) \> \> \\ 
     \> \sendmessageleft*[2.5cm]{\texttt{d($g^b$)}} \> \\ 
     g^b = Open(c(g^b), d(g^b)) \> \> \\
     g^{ab} = (g^b)^a \> \> \\
     AV_B^{\prime} = \texttt{G}(A, s, g^b, c(g^b), N_A)}
\end{pchstack}
The security of the above protocol comes as a corollary of the SK security of the generic 4-pass KA-based construction:
\begin{proposition}
    The 4-pass DH-based protocol is $(\advantage{DDH}{DH}[(\mathcal{B})] + 2\cdot \alpha)$-SK secure in the UM,
    where $\alpha$ is the MT emulation margin proven in Proposition \ref{CS_MT_security}.
\end{proposition}
\begin{proof}
    A particular application of Proposition \ref{gen_4_pass_KA_sk_security} with $\epsilon$ as given in Proposition \ref{2pass_dh_am} regarding SK security in the AM of the 2-pass.
\end{proof}

\section{KEM-based protocols}

When working with Key Encapsulation Mechanisms, two important differences must be highlighted:
\begin{enumerate}
    \item KEM mechanisms are, in general, not contributory, as opposed to KA mechanisms. This means that the final shared key does not come from a mutual contribution of both ends, but is unilaterally generated by one end and transmitted to the other. This paradigm, to which KEM schemes belong, is usually referred to as \textit{Key Transport} protocols.
    \item The public values that will be exchanged in the protocol are not independent. That is, it will not be formed by two public keys generated by both users, but formed by a public key and an encapsulation, whose value depends on the public key value.
\end{enumerate}
These two differences will drive the security analyses made on KEM-based protocols, and its contrasts with their KA-based counterparts.

\subsection{A KEM-based SK-secure protocol in the AM: Choice of primitive \label{kem_2_pass_sk_secure}}

The 2-pass KEM-based protocol is, in appearance, a drop-in replacement of its KA-based counterpart, in the sense that the KA instances were substituted by \textit{Encaps} and \textit{Decaps} procedures:

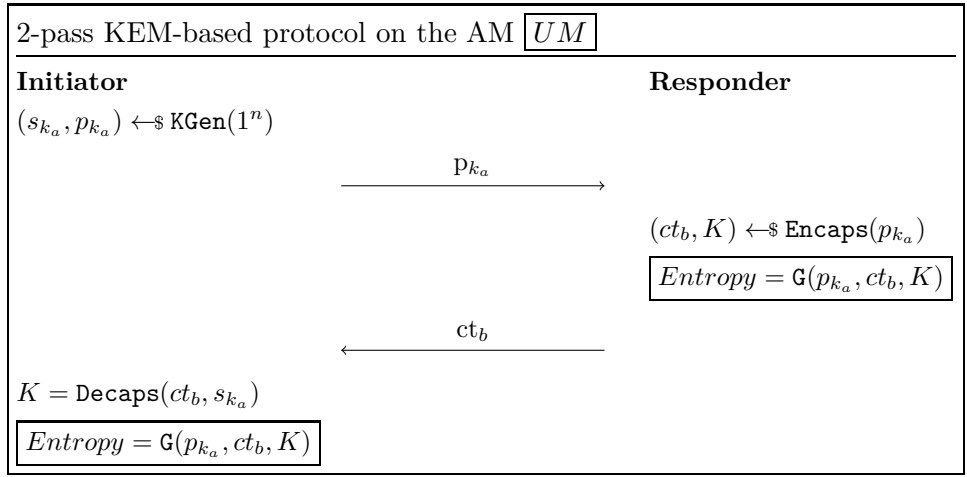
\begin{pchstack}[boxed, center, space=1em]
    \procedure{2-pass KEM-based protocol on the AM \pcbox{UM}}{%
     \textbf{Alice} \> \> \textbf{Bob} \\
     (s_{k_a}, p_{k_a}) \sample \texttt{KGen} (\secparam) \> \> \\
     \> \sendmessageright{top=\text{(p$_{k_a}$, s)}} \> \\
     \> \> (ct_b, K) \sample \texttt{Encaps}(p_{k_a}) \\
     \> \> \pcbox{E = \texttt{G}(B, p_{k_a}, s,  ct_b, K)} \\
     \> \sendmessageleft{top=\text{(ct$_b$, s)}} \> \\ 
     K = \texttt{Decaps}(ct_b, s_{k_a}) \> \> \\
    \pcbox{E = \texttt{G}(B, p_{k_a}, s,  ct_b, K)} \> \>}
\end{pchstack}

\begin{itemize}
    \item Alice generates a cryptographic key pair ($s_{k_a}$, $p_{k_a}$), by means of the KGen() function of the specific KEM selected. Then, the public value $p_{k_a}$ is sent to Bob, along with the session identifier.
    \item  Bob, upon reception of Alice's public key $p_{k_a}$, executes the encapsulation function on the public key received. The execution of this function yields a secret key K, and an encapsulation $ct_b$ of the secret that conforms the generated shared secret. Then, Bob sends the aforementioned encapsulation.
    \item Alice executes the decapsulation function, with inputs the encapsulation $ct_b$ received and its own secret key, to derive the shared secret K.
\end{itemize}
\subsubsection{Security}
Under the AM model, that is, under the assumption that the integrity of each message received is intact, the protocol is SK-secure, as proven in \cite{10.1007/978-3-031-35486-1_24}:

\begin{proposition}\cite{10.1007/978-3-031-35486-1_24}
    Let $\adv$ be an adversary against the SK-security of the above protocol, which interacts with at most $q$ sessions for each pair of $n_P$ parties. Then, the 2-pass KEM-based protocol is $\epsilon$-SK-secure in the AM, with 
    $$ \epsilon \leq q\cdot n_P^2 \cdot \advantage{CPA}{KEM}$$
\end{proposition}
\begin{proof}
    See \cite{10.1007/978-3-031-35486-1_24}.
\end{proof}

Under our UM model, a number of different considerations arise regarding the protocol's practical security: the attacker's possibility to generate the same shared secret on both ends, a replica for KEM-based protocols of the attack against 2-pass KA-based protocol, and a combination of them.
\paragraph*{Attack on same key on both ends}
Despite its clear similarities to the 2-pass KA-based protocol, a MitM attacker not only has the ability to establish a shared secret with each party but also, in general, the ability to generate the exact same shared secret on both, and thus all three parties possess the same shared secret value. This is a direct consequence of the non-contributory nature of Key Encapsulation Mechanisms. The attack would go as follows:
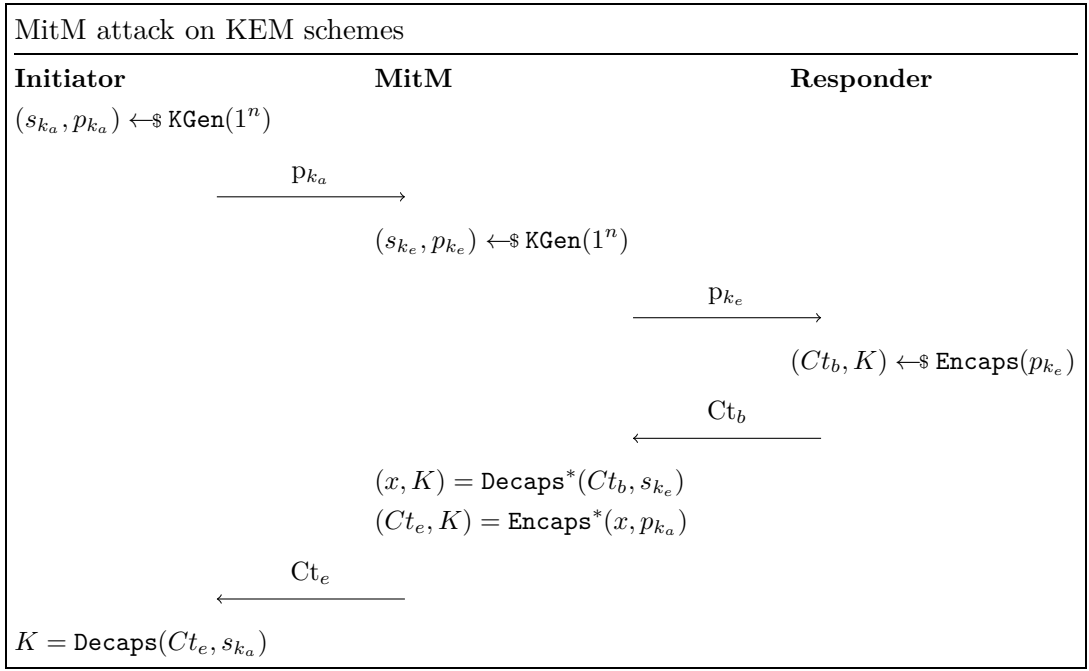
\begin{pchstack}[boxed, center, space=1em]
    \procedure[colspace=-1cm]{MitM attack on KEM schemes}{%
     \textbf{Alice}  \> \> \textbf{Mallory} \> \> \textbf{Bob} \\
     (s_{k_a}, p_{k_a}) \sample \texttt{KGen} (\secparam) \> \> \> \> \\
     \> \sendmessageright*[2.5cm]{\text{p$_{k_a}$}} \> \> \> \\
     \> \> (s_{k_e}, p_{k_e}) \sample \texttt{KGen} (\secparam) \> \> \\ 
     \> \> \> \sendmessageright*[2.5cm]{\text{p$_{k_e}$}} \> \\
     \> \> \> \> (Ct_b, K) \sample \texttt{Encaps}(p_{k_e}) \\
     \> \> \> \sendmessageleft*[2.5cm]{\text{Ct$_b$}} \> \\
     \> \> (x, K) = \texttt{Decaps$^*$}(Ct_b, s_{k_e}) \> \> \\
     \> \> (Ct_e, K) = \texttt{Encaps$^*$}(x, p_{k_a}) \> \> \\
     \> \sendmessageleft*[2.5cm]{\text{Ct$_e$}} \> \> \> \\
     K = \texttt{Decaps}(Ct_e, s_{k_a}) \> \> \> \>}
\end{pchstack}
\begin{itemize}
    \item Alice generates its KEM key pair ($s_{k_{a}}$, $p_{k_{a}}$) and sends its public key to Bob.
    \item Mallory intercepts the public key value sent, generates another key pair ($s_{k_{e}}$, $p_{k_{e}}$) and substitutes the public key value sent by Alice with its own public key $p_{k_{e}}$.
    \item Bob executes the encapsulation function with the public key value received. At this moment, all information about the key is established. Bob sends the encapsulation of the secret required to derive the shared key.
    \item Mallory intercepts the encapsulation sent and, as it has been encapsulated using its public key value, decapsulates it using the secret key. Through this process, the attacker is able to learn the secret $x$ that was encapsulated, and that forms the shared key. Therefore, the attacker simply re-encapsulates this exact same value with Alice's public key, and sends this encapsulation Ct$_e$ to them.
    \item Alice executes the decapsulation function on the encapsulation received, and its own secret key, to derive the same shared key \textit{K} as Bob and Mallory.
\end{itemize}
This possibility highlights the requirement of the session AV to be contributed by the public values involved in the key derivation, as this addition is all that is required to  thwart the attack, for deterministic-based KEMs.

The previous attack is carried on by slightly tweaking the \textit{Encaps} and \textit{Decaps} definition provided, on the attacker's side. The \textit{Encaps}* function takes also as argument the secret value to be encapsulated. Meanwhile, the \textit{Decaps}* function also outputs the secret value that was encapsulated.
\paragraph*{Replica attack on KA-based protocol \label{4.1.1.2}}
When the above protocol's session AV is contributed with public values involved in the exchange, the attack detailed above does not, in general, apply (see below). Nevertheless, the attack on the session AV shown in Section \ref{3.1.1} still applies, as a commitment value is not present in the 2-pass KEM-based protocol either:
\begin{itemize}
    \item When Alice sends a public key value $p_{k_{a}}$, the attacker intercepts it and transmits instead the public key value $p_{k_{e}}$ of a newly generated key pair ($s_{k_{e}}$, $p_{k_{e}}$). Then, Bob will execute the \textit{Encaps} function over the fraudulent public key $p_{k_{e}}$, generating a secret $K_{be}$, and the corresponding session AV, based on this secret, the encapsulation $Ct_{b}$ generated and the fraudulent public key value $p_{k_e}$. Then, Bob sends this encapsulation to Alice.
    \item Mallory would then need to intercept the encapsulation value $Ct_{b}$ and transmit instead an encapsulation value $Ct_{e_{i_{0}}}$ of an execution of the \textit{Encaps} function on Alice's public key $p_{k_{a}}$. But this value must be such that the session AV generated by the shared secret $K_{ae}$ and the public values $Ct_{e}$ and $p_{k_a}$ its the same as the session AV generated by Bob (which Mallory is able to calculate, as possesses all elements involved in its generation). Therefore, the attacker would generate enough encapsulation values $Ct_{e_{i}}$, looping through \textit{Encaps} executions, until one encapsulation $Ct_{e_{i_{0}}}$ satisfies the required condition.
\end{itemize}
\begin{pchstack}[boxed, center, space=1em]
    \procedure[colspace=-1.75cm]{MitM attack on 2-pass KEM-based protocol}{%
     \textbf{Alice}  \> \> \textbf{Mallory} \> \> \textbf{Bob} \\
     (s_{k_a}, p_{k_a}) \sample \texttt{KGen} (\secparam) \> \> \> \> \\
     \> \sendmessageright{top=\text{(p$_{k_a}$, s}} \> \> \> \\
     \> \> (s_{k_e}, p_{k_e}) \sample \texttt{KGen} (\secparam) \> \> \\ 
     \> \> \> \sendmessageright{top=\text{(p$_{k_e}$, s}} \> \\
     \> \> \> \> (Ct_b, K_{be}) \sample \texttt{Encaps}(p_{k_e}) \\
     \> \> \> \>  E_{be} = \texttt{G}(B, p_{k_e}, s,  Ct_b, K_{be}) \\
     \> \> \> \sendmessageleft{top=\text{(Ct$_b$, s}} \> \\
     \> \> K_{eb} = \texttt{Decaps}(Ct_b, s_{k_e}) \> \> \\
     \> \> E_{eb} = \texttt{G}(B, p_{k_e}, s,  Ct_b, K_{eb}) \> \> \\
     \> \> \texttt{While $E_{ea} \neq E_{eb}$ do:} \> \> \\
     \> \> \quad (Ct_e, K_{ea}) \sample \texttt{Encaps}(p_{k_a}) \> \> \\
     \> \> \quad E_{ea} = \texttt{G}(B, p_{k_a}, s,  Ct_e, K_{ea}) \> \> \\
     \> \sendmessageleft{top=\text{(Ct$_e$, s}} \> \> \> \\
     K_{ae} = \texttt{Decaps}(Ct_e, s_{k_a}) \> \> \> \> \\
     E_{ae} = \texttt{G}(B, p_{k_a}, s,  Ct_e, K_{ae}) \> \> \> \>
     }
\end{pchstack}

As in Section \ref{3.1.1}, the above diagram provides the following lemma:
\begin{proposition}
The 2-pass KEM-based protocol has SK-security in the UM model, regardless of the KEM scheme employed, bounded by the following value:
\begin{align*}
    \advantage{\texttt{Key-Ind}}{\pi_{UM}}[(\adv)] \geq 1 - \left(1 - \frac{1}{2^l}\right)^{\min\{q, |Y|\}}
\end{align*}
where $l$ is the output length of the AV random oracle, $q$ is the maximum number of distinct queries the adversary is able to make against the AV random oracle and $Y: = \{(ct, K) :  (ct, K) \sample \texttt{Encaps}(\pk_p)\}$ is the set of all possible \textit{Encaps} results of a certain public key $\pk_p$.
\end{proposition}
\begin{proof}
The above diagram can be replicated exactly as a routine for an adversary $\adv$ in the UM unless it is unable to generate an AV collision between $E_{ab}$ and $E_{eb}$.

Note the advantage of an adversary in generating said collision is defined exactly as $\advantage{\texttt{combined}}{\texttt{CHF}}[(l, \mathcal{Y})]$, where $\mathcal{Y}$ is an algorithm that returns values from the domain $Y$ by executing the \textit{Encaps} procedure with $\pk_p$. 

Therefore, since $Y$ does not impose any cryptographic restriction that the pairs $(ct, K)$ must verify, we have $\advantage{\texttt{sample}}{Y}[(\mathcal{A})] = 1$. Consequently, the adversary can construct a set $Y_q$ of $q$ different values $((\pk_i, K_i))$.

Since the $SK$ advantage can be lower bounded by the advantage of any plausible attack, we have that:
\begin{align*}
    \advantage{\texttt{Key-Ind}}{\pi_{UM}}[(\adv)] \geq \advantage{\texttt{combined}}{\texttt{CHF}}[(\adv, l, \mathcal{Y})] = 1 - \left(1 - \frac{1}{2^l}\right)^{\min\{q, |Y|\}}
\end{align*}
where the last equality is provided by Proposition \ref{combined_adv_with_Y}.
\end{proof}

As noted in the KA setting, while we do not claim that this protocol is always insecure, it is important to realize that, depending upon the capabilities of the adversary in terms of the maximum number of queries $q$, the protocol could become insecure. 

This realization motivates the need to come up with more robust protocols in which the security is guaranteed, in a KEM-based setting.

\paragraph*{Combination under probabilistic PKE algorithm \label{4.1.1.3}}
In the above setting, if the underlying PKE is deterministic, an attacker is not able to generate the same key on both ends. This is due to the fact that, in case that the key value is fixed to be the same on both ends, the attacker would need a different encapsulation value generated on each iteration, over the same secret \textit{t} that generates the shared key, and this cannot be achieved on a deterministic PKE. 

On the other hand, if the underlying PKE scheme is indeed probabilistic and has not undergone a de-randomization procedure, the two attacks can happen simultaneously, as each execution of the encapsulation procedure on the same secret value yields different results.

\subsection{A KEM-based SK-secure protocol in the UM: Compiler application}
The above protocol is shown to be secure on the AM, if the KEM selected is IND-CPA secure. To avoid the vulnerabilities that the protocol presented on the UM, the application of a compiler based on the MT-authenticator defined on section can be applied, yielding the following (unoptimized) protocol:
\begin{pchstack}[boxed, center, space=1em]
    \procedure{6-pass KEM-based protocol (unoptimized)}{%
     \textbf{Alice} \> \> \textbf{Bob} \\
     (\pk,\sk) \sample \texttt{KGen} (\secparam) \> \> \\
     m \sample \texttt{\{0,1\}$^n$} \> \> \\
    (c(m), d(m)) \sample \texttt{Com}(m) \> \> \\
     \> \sendmessageright{top=\texttt{(c(m), s, \pk)}} \> \\
     \> \>  N_B \sample \texttt{\{0,1\}$^n$} \\
     \> \sendmessageleft{top=\texttt{$N_B$}} \> \\
     AV_A = \texttt{G}(B, c(m), s, \pk, N_B, m) \> \> \\
     \> \sendmessageright{top=\texttt{d(m)}} \> \\
     \> \> m = Open(c(m), d(m)) \\
     \> \> AV_A = \texttt{G}(B, c(m), s, \pk, N_B, m) \> \> \pclb
     \pcintertext[dotted]{$\lambda_I - \lambda_R$ division} \\
     \> \> (ct, K) \sample \texttt{Encaps}(\pk) \\ 
     \> \> (c(ct), d(ct)) \sample \texttt{Com}(ct) \\
     \> \sendmessageleft{top=\texttt{(c(ct), s)}} \> \\ 
      N_A \sample \texttt{\{0,1\}$^n$} \> \> \\
     \> \sendmessageright{top=\texttt{$N_A$}} \> \\ 
     \> \> AV_B = \texttt{G}(A, c(ct), s, N_A, ct) \> \> \\ 
     \> \sendmessageleft{top=\texttt{d(ct)}} \> \\ 
     ct = Open(c(ct), d(ct)) \> \> \\
     K = \texttt{Decaps}(\sk, ct) \> \> \\
     AV_B = \texttt{G}(A, c(ct), s, N_A, ct)}
\end{pchstack}
where each of the two separated blocks represent a (almost) canonical application of the CS-based MT-authenticator to one message of the 2-pass KA-based protocol.

As in Section, the only conceptual difference is that the MT-authenticators provide an additional value in the first message exchanged, i.e. the session identifier, which is also appended to the AV calculation.

To be precise, in this case the first MT-authenticator exchange includes two additional values in the first message: a session identifier and an ephemeral public key, while the second exchange includes just one additional value in the first message.

Note that, while the public key added on the first message is later employed on constructing the encapsulation value, it remains independent towards all the information exchanged in the first MT-authenticator exchange and therefore it does not disrupt the security of the authenticator.

\subsubsection{Security}
The security of the above protocol is a direct application of propositions \ref{CS_MT_security}, \ref{gen_MT_compiler} and \ref{ThCompilerSK}, along with the security of the 2-pass KEM-based protocol in the AM. Formally:

\begin{proposition}\label{gen_4_pass_KEM_sk_security}
    Let the 2-pass KEM-based protocol be a $\epsilon$-SK-secure protocol in the AM. The 6-pass KEM-based protocol is $(\epsilon + 2\cdot \alpha)$-SK secure in the UM,
    where $\alpha$ is the MT emulation margin proven in Proposition \ref{CS_MT_security}.
\end{proposition}
\begin{proof}
The CS-based MT-authenticator, denoted as $\lambda$,  $\alpha$-emulates the MT-authenticator by Proposition \ref{CS_MT_security}.

Then, Proposition \ref{gen_MT_compiler} proves that the sequence $\Lambda = (\lambda, \lambda)$ defines a compiler $C_{\Lambda}$ such that, for any protocol $\Pi$ in the AM, $C_{\Lambda}(\Pi)$ $2\cdot \alpha$ emulates $\Pi$ in the UM.

Then, by Theorem $\ref{ThCompilerSK}$, $C_{\Lambda}(\Pi)$ is an $\epsilon'$-SK-secure protocol in the UM, where $\epsilon' := \epsilon + 2 \cdot\alpha$.
\end{proof}

The only modification over a canonical application of the CS-based compiler is the inclusion of the $\pk$ value in the first message of the first exchange, and the inclusion of the session entropy in the first message of the two exchanges. As discussed above, this is perfectly valid, since we can add any value transmitted in the first of the CS MT-authenticator without any repercussion to the validity of the authenticator, if they remain independent to the rest of the values transmitted on the exchange, and are part of the AV calculation.
 
This protocol provides mutual authentication between parties, as the responder authenticates value $E_B$ against Alice, and Alice does the same against Bob with value $E_A$.

Note that this specification does not address how the session identifier is actually generated. This could just be an unique identifier, formalized by each protocol specification. 
Alternatively, the session identifier could also be formed by values exchanged in the protocol run, if parties can verify that no incomplete sessions between the same parties have the same session identifier. For example, $s = N_A || N_B$ could be a selection for the session value, as both $N_A$ and $N_B$ are random $n$-bit strings generated ephemerally on each protocol run.

\subsubsection{Protocol optimization}
The above protocol represents a canonical application of the CS-based compiler: The first three messages are the MT-authentication of the message $m$, along with the inclusion of extra element $\pk$ and the session identifier $s$ sent on the first message. 

The second part of the exchange amounts to the MT-authentication of element \textit{ct} sent by the responder, and we define this exchange as $\lambda_R$. Then, by sending \textit{in parallel}:
\begin{itemize}
    \item the 2nd message of $\lambda_I$ with the 1st message of $\lambda_R$ as the second message of the protocol, and
    \item the 3rd message of $\lambda_I$ with the 2nd of $\lambda_R$ as the 3rd message of the protocol.
\end{itemize}
we have the optimized 4-pass protocol:

\begin{pchstack}[boxed, center, space=1em]
    \procedure{4-pass KEM-based protocol (optimized)}{%
     \textbf{Alice} \> \> \textbf{Bob} \\
     (\pk,\sk) \sample \texttt{KGen} (\secparam) \> \> \\
     m \sample \texttt{\{0,1\}$^n$} \> \> \\
    (d(m), c(m)) \sample \texttt{Com}(m) \> \> \\
     \> \sendmessageright*[2.5cm]{\texttt{(\pk , s, c(m))}} \> \\
     \> \>  N_B \sample \texttt{\{0,1\}$^n$} \\
     \> \> (ct, K) \sample \texttt{Encaps}(\pk) \\ 
     \> \> (d(ct), c(ct)) \sample \texttt{Com} (ct) \\
     \> \sendmessageleft*[2.5cm]{\texttt{(c(ct), $N_B$)}} \> \\ 
      N_A \sample \texttt{\{0,1\}$^n$} \> \> \\
     E_A = \texttt{G}(B, m, c(m), N_B, \pk , s) \> \> \\ 
     \> \sendmessageright*[2.5cm]{\texttt{(d(m), $N_A$)}} \> \\
     \> \> m = Open(c(m), d(m)) \\
     \> \> E_A^{\prime} = \texttt{G}(B, m, c(m), N_B, \pk , s) \> \> \\ 
     \> \> E_B = \texttt{G}(A, ct, c(ct), N_A) \> \> \\ 
     \> \sendmessageleft*[2.5cm]{\texttt{d(ct)}} \> \\ 
     ct = Open(c(ct), d(ct)) \> \> \\
     K = \texttt{Decaps}(\sk, ct) \> \> \\
     E_B^{\prime} = \texttt{G}(A, ct, c(ct), N_A)}
\end{pchstack}
\begin{itemize}
    \item First, Alice generates a key pair $(\sk, \pk)$ and a random value \textit{m} uniformly, and sends the public key value $\pk$ and a commitment \textit{c(m)} of \textit{m}, along with the session identifier $s$.
    \item Bob, upon reception of the public key value $\pk$, executes the \textit{Encaps} function on the public key received, generating the secret key value \textit{K} and an encapsulation \textit{ct}. Then, generates a commitment \textit{c(ct)} of this encapsulation value and sends it, along with a random challenge $N_B$.
    \item The inititator, upon reception of the commitment value \textit{c(ct)}, generates its own challenge $N_A$ and sends it along with the opening value \textit{d(m)} of \textit{m}.
    \item Bob verifies the commitment received at the beginning with the opening value received in this iteration and the actual value committed and, upon successful verification, calculates its session AV and sends the opening value \textit{d(ct)} of the encapsulation value \textit{ct} committed before.
    \item Alice verifies that the value committed matches the commitment stored, and then proceeds to derive the shared secret key using the \textit{Decaps} function, and its corresponding session AV.
\end{itemize}

The security of the above optimized protocol comes as a corollary of the SK security of the unoptimized 6-pass KEM-based construction:
\begin{proposition}
    The 4-pass KEM-based protocol is $(\advantage{IND-CPA}{KEM}[(\mathcal{B})] + 2\cdot \alpha)$-SK secure in the UM,
    where $\alpha$ is the MT emulation margin proven in Proposition \ref{CS_MT_security}.
\end{proposition}
\begin{proof}
    Application of Proposition \ref{gen_4_pass_KEM_sk_security}, as messages have only been sent in parallel.
\end{proof}

\section{Conclusions}
First, we have defined an alternative \textit{Unauthenticated Model}, following the same ideas first explained in \cite{cryptoeprint:1998/009}, but providing a different source of authentication, not based on the assumption of a initial safe exchange of cryptographic material, but rather a final safe verification phase of protocol-generated elements. This modification is consistent with the use of key exchange protocols under different security scenarios from the ones provided by the original authentication source. With the appropriate modification, we have shown that all relevant results from their model still apply to ours.

We have shown how to construct secure KEM-based protocols over this new UM, which requires to consider the plausible presence of MitM attackers. We have done so by first studying how equivalent KA-based protocols would be constructed in this model, and the difficulties that naturally arise from the differences between standard KA constructions like the Diffie-Hellman paradigm and KEM constructions.

The protocols constructed follow on the steps taken in the KA settings, that is, generating an additional value that \textit{resumes} the session exchange between the two parties, and forcing any attacker to commit to a specific value before it has all the necessary information to calculate such value, so its presence is detected by inconsistencies in this value.

For the 4-pass protocol, it is indeed possible for the responder to generate a commitment of its public value, but it is required to show an actual separation, in terms of Alice's interaction, between the commitment value and the actual value. Therefore, a \textit{'proof of identity'} is generated, by means of yet another commitment, this time of a random value large enough to ensure Alice's identity. With this measures, we show the impossibility of a MitM attack and, consequently, the security of the protocol over our UM.

These protocols, based on the authentication differences between the model started by \cite{cryptoeprint:1998/009} and followed by \cite{10.1007/978-3-031-35486-1_24}, account for fewer public key operations to derive the authentication source. The authentication source provided in this new model is intended to correspond with other practical applications to communication protocols.

\bibliography{lib_published.bib}

\newcommand{\etalchar}[1]{$^{#1}$}
\begin{thebibliography}{BdKM23}

\bibitem[AGKS05]{10.1007/11426639_8}
Masayuki Abe, Rosario Gennaro, Kaoru Kurosawa, and Victor Shoup.
\newblock {Tag-KEM/DEM: A New Framework for Hybrid Encryption and A New Analysis of Kurosawa-Desmedt KEM}.
\newblock In Ronald Cramer, editor, {\em Advances in Cryptology -- EUROCRYPT 2005}, pages 128--146, Berlin, Heidelberg, 2005. Springer Berlin Heidelberg.

\bibitem[BCK98]{cryptoeprint:1998/009}
Mihir Bellare, Ran Canetti, and Hugo Krawczyk.
\newblock {A Modular Approach to the Design and Analysis of Authentication and Key Exchange Protocols}.
\newblock Cryptology ePrint Archive, Paper 1998/009, 1998.
\newblock \url{https://eprint.iacr.org/1998/009}.

\bibitem[BdKM23]{10.1007/978-3-031-35486-1_24}
Colin Boyd, Bor de~Kock, and Lise Millerjord.
\newblock {Modular Design of KEM-Based Authenticated Key Exchange}.
\newblock In {\em Information Security and Privacy: 28th Australasian Conference, ACISP 2023, Brisbane, QLD, Australia, July 5–7, 2023, Proceedings}, page 553–579, Berlin, Heidelberg, 2023. Springer-Verlag.

\bibitem[BFG{\etalchar{+}}20]{10.1007/978-3-030-81652-0_16}
Jacqueline Brendel, Marc Fischlin, Felix G\"{u}nther, Christian Janson, and Douglas Stebila.
\newblock {Towards Post-Quantum Security for Signal’s X3DH Handshake}.
\newblock In {\em Selected Areas in Cryptography: 27th International Conference, Halifax, NS, Canada (Virtual Event), October 21-23, 2020, Revised Selected Papers}, page 404–430, Berlin, Heidelberg, 2020. Springer-Verlag.

\bibitem[CK01]{10.1007/3-540-44987-6_28}
Ran Canetti and Hugo Krawczyk.
\newblock {Analysis of Key-Exchange Protocols and Their Use for Building Secure Channels}.
\newblock In Birgit Pfitzmann, editor, {\em Advances in Cryptology - EUROCRYPT 2001}, pages 453--474, Berlin, Heidelberg, 2001. Springer Berlin Heidelberg.

\bibitem[CS03]{doi:10.1137/S0097539702403773}
Ronald Cramer and Victor Shoup.
\newblock {Design and Analysis of Practical Public-Key Encryption Schemes Secure against Adaptive Chosen Ciphertext Attack}.
\newblock {\em SIAM Journal on Computing}, 33(1):167--226, 2003.

\bibitem[CV11]{Canetti2011}
Ran Canetti and Mayank Varia.
\newblock Decisional diffie--hellman problem.
\newblock In Henk C.~A. van Tilborg and Sushil Jajodia, editors, {\em Encyclopedia of Cryptography and Security}, pages 316--319. Springer US, Boston, MA, 2011.

\bibitem[Den03]{10.1007/978-3-540-40974-8_12}
Alexander~W. Dent.
\newblock {A Designer's Guide to KEMs}.
\newblock In Kenneth~G. Paterson, editor, {\em Cryptography and Coding}, pages 133--151, Berlin, Heidelberg, 2003. Springer Berlin Heidelberg.

\bibitem[HHK17]{10.1007/978-3-319-70500-2_12}
Dennis Hofheinz, Kathrin H{\"o}velmanns, and Eike Kiltz.
\newblock {A Modular Analysis of the Fujisaki-Okamoto Transformation}.
\newblock In Yael Kalai and Leonid Reyzin, editors, {\em Theory of Cryptography}, pages 341--371, Cham, 2017. Springer International Publishing.

\bibitem[Kra05]{10.1007/11535218_33}
Hugo Krawczyk.
\newblock {HMQV: A High-Performance Secure Diffie-Hellman} protocol.
\newblock In Victor Shoup, editor, {\em Advances in Cryptology -- CRYPTO 2005}, pages 546--566, Berlin, Heidelberg, 2005. Springer Berlin Heidelberg.

\bibitem[LLM07]{10.1007/978-3-540-75670-5_1}
Brian LaMacchia, Kristin Lauter, and Anton Mityagin.
\newblock {Stronger Security of Authenticated Key Exchange}.
\newblock In Willy Susilo, Joseph~K. Liu, and Yi~Mu, editors, {\em Provable Security}, pages 1--16, Berlin, Heidelberg, 2007. Springer Berlin Heidelberg.

\bibitem[Sho97]{Shor_1997}
Peter~W. Shor.
\newblock {Polynomial-Time Algorithms for Prime Factorization and Discrete Logarithms on a Quantum Computer}.
\newblock {\em {SIAM} Journal on Computing}, 26(5):1484--1509, oct 1997.

\bibitem[Sho04]{cryptoeprint:2004/332}
Victor Shoup.
\newblock Sequences of games: a tool for taming complexity in security proofs.
\newblock Cryptology ePrint Archive, Paper 2004/332, 2004.
\newblock \url{https://eprint.iacr.org/2004/332}.

\end{thebibliography}
\bibliographystyle{alpha}

\nocite{10.1007/978-3-030-81652-0_16}
\nocite{10.1007/11426639_8}

\end{document}